\documentclass[runningheads]{llncs}

\usepackage{booktabs}   
\usepackage{subcaption} 
\usepackage{cmap}

\usepackage{shuffle}
\usepackage{syntax}

\usepackage[square,sort,comma,numbers]{natbib}
\usepackage{url}
\usepackage{graphicx}
\usepackage{amssymb} 
\usepackage{amsfonts}
\usepackage{calligra}

\usepackage{etoolbox}
\usepackage{mathpartir}
\usepackage{mathtools}
\usepackage{refcount}
\usepackage{stmaryrd}
\usepackage{enumitem}
\usepackage{listings}
\usepackage{caption}
\usepackage{flushend}
\usepackage{marvosym}

\usepackage{xcolor}

\usepackage{algorithm}
\usepackage[noend]{algpseudocode}

\usepackage{comment}  

\definecolor{bluekeywords}{rgb}{0,0,1}
\definecolor{greencomments}{rgb}{0,0.5,0}
\definecolor{redstrings}{rgb}{0.64,0.08,0.08}
\definecolor{xmlcomments}{rgb}{0.5,0.5,0.5}
\definecolor{types}{rgb}{0.17,0.57,0.68}

\usepackage{listings}
\usepackage{array}
\newcolumntype{L}[1]{>{\raggedright\let\newline\\\arraybackslash\hspace{0pt}}m{#1}}
\newcolumntype{C}[1]{>{\centering\let\newline\\\arraybackslash\hspace{0pt}}m{#1}}
\newcolumntype{R}[1]{>{\raggedleft\let\newline\\\arraybackslash\hspace{0pt}}m{#1}}

\usepackage{pgf}
\usepackage{tikz}
\usetikzlibrary{tikzmark}
\usetikzlibrary{calc,arrows,automata,fit}
\usetikzlibrary{positioning,shapes,shadows}
\usetikzlibrary{decorations.pathmorphing,snakes}

\usepackage{mathdots}

\usepackage{algpseudocode}
\makeatletter
\renewcommand{\ALG@beginalgorithmic}{\footnotesize}
\makeatother
\usepackage{eqnarray}
\usepackage{alltt}
\usepackage{textcomp}
\usepackage{wrapfig,lipsum}
\usepackage{mathsprograms}

\usepackage{listings}
\lstdefinelanguage{JavaScript}{%
  keywords = { async, await, break, case, catch, class, const, continue, debugger, default, delete, do, each, else, export, finally, for, function, if, import, in, instanceof, let, new, of, return, switch, this, throw, try, typeof, var, void, while, with, yield },
  morecomment = [l]{//},
  morecomment = [s]{/*}{*/},
  morestring  = [b]',
  morestring  = [b]",
  sensitive   = true,
}

\lstdefinelanguage{Java10}{
  language      = Java,
  morekeywords  ={ var },
}

\usepackage{float}

\begin{document}

%
%
\title{Checking Robustness Between Weak Transactional Consistency Models\thanks{This work is supported in part by the European Research Council (ERC) under the European Union's Horizon 2020 research and innovation programme (grant agreement No 678177).}\vspace{-10pt}}
\author{Sidi Mohamed Beillahi\textsuperscript{(\Letter)} \and
Ahmed Bouajjani \and
Constantin Enea}
\authorrunning{S.M. Beillahi, A. Bouajjani, and C. Enea.}
%
\institute{Universit\'e de Paris, IRIF, CNRS, Paris, France,
\email{\{beillahi,abou,cenea\}@irif.fr}}
\maketitle              
\begin{abstract}

Concurrent accesses to databases are typically encapsulated in transactions in order to enable isolation from other concurrent computations and resilience to failures. Modern databases provide transactions with various semantics corresponding to different trade-offs between consistency and availability. Since a weaker consistency model provides better performance, an important issue is investigating  the weakest level of consistency needed by a given program (to satisfy its specification). 
As a way of dealing with this issue, we investigate the problem of checking whether a given program has the same set of behaviors when replacing a consistency model with a weaker one. 
This property known as \emph{robustness} generally implies that any specification of the program is preserved when weakening the consistency.
We focus on the robustness problem for consistency models which are weaker than standard serializability, namely, causal consistency, prefix consistency, and snapshot isolation. 
We show that checking robustness between these models is polynomial time reducible to a state reachability problem under serializability.
We use this reduction to also derive a pragmatic proof technique based on Lipton's reduction theory that allows to prove programs robust.
We have applied our techniques to several challenging applications drawn from the literature of distributed systems and databases. 
    \end{abstract}
  \vspace{-15pt}
\vspace{-15pt}

\section{Introduction}
\label{sec:intro}
Concurrent accesses to databases are typically encapsulated in transactions in order to enable isolation from other concurrent computations and resilience to failures. Modern databases provide transactions with various semantics corresponding to different tradeoffs between consistency and availability. The strongest consistency level is achieved with \emph{serializable} transactions~\cite{DBLP:journals/jacm/Papadimitriou79b} whose outcome in concurrent executions is the same as if the transactions were executed atomically in some order. Since serializability (\serc{}) carries a significant penalty on availability, modern databases often provide weaker consistency models, e.g., causal consistency (\ccc{})~\cite{DBLP:journals/cacm/Lamport78}, prefix consistency (\pcc{})~\cite{DBLP:conf/ecoop/BurckhardtLPF15,DBLP:conf/concur/Cerone0G15}, and snapshot isolation (\sic{})~\cite{DBLP:conf/sigmod/BerensonBGMOO95}. Causal consistency requires that if a transaction $\tr_1$ ``affects'' another transaction $\tr_2$, e.g., $\tr_1$ executes before $\tr_2$ in the same session or $\tr_2$ reads a value written by $\tr_1$, then the updates in these two transactions are observed by any other transaction in this order. Concurrent transactions, which are not causally related to each other, can be observed in different orders, leading to behaviors that are not possible under \serc{}. Prefix consistency requires that there is a total commit order between all the transactions such that each transaction observes all the updates in a prefix of this sequence (\pcc{} is stronger than \ccc{}). Two transactions can observe the \emph{same} prefix, which leads to behaviors that are not admitted by \serc{}. Snapshot isolation further requires that two different transactions observe different prefixes if they both write to a common variable.

Since a weaker consistency model provides better performance, an important issue is identifying the \emph{weakest} level of consistency needed by a program (to satisfy its specification). One way to tackle this issue is checking whether a program $P$ designed under a consistency model $S$ has the same behaviors when run under a weaker consistency model $W$. This property of a program is generally known as \emph{robustness} against substituting $S$ with $W$. 
It implies that any specification of $P$ is preserved when weakening the consistency model (from $S$ to $W$). Preserving any specification is convenient since specifications are rarely present in practice.

The problem of checking robustness for a given program has been investigated in several recent works, but only when the stronger model ($S$) is \serc{}, e.g.,~\cite{DBLP:conf/cav/BeillahiBE19,DBLP:conf/concur/BeillahiBE19,DBLP:conf/pldi/BrutschyD0V18,DBLP:journals/jacm/CeroneG18,DBLP:conf/concur/0002G16,DBLP:conf/concur/NagarJ18}, or sequential consistency in the non-transactional case, e.g.~\cite{DBLP:conf/pldi/LahavM19,DBLP:conf/esop/BouajjaniDM13,DBLP:conf/icalp/DerevenetcM14}.
However, 
there is a large class of specifications that can be implemented even in the presence of ``anomalies'', i.e., behaviors which are not admitted under \serc{} (see~\cite{DBLP:conf/concur/ShapiroAP16} for a discussion). In this context, an important question is whether a certain implementation (program) is robust against substituting a weak consistency model, e.g., \sic{}, with a weaker one, e.g., \ccc{}. 

In this paper, we consider the sequence of increasingly strong consistency models mentioned above, \ccc{}, \pcc{}, and \sic{}, and investigate the problem of checking robustness for a given program against weakening the consistency model to one in this range. 
We study the asymptotic complexity of this problem and propose effective techniques for establishing robustness based on abstraction. 
There are two important cases to consider: robustness against substituting \sic{} with \pcc{} and \pcc{} with \ccc{}, respectively. Robustness against substituting \sic{} with \ccc{} can be obtained as the conjunction of these two cases. 

In the first case (\sic{} vs \pcc{}), checking robustness for a program $P$ 
is reduced to a reachability (assertion checking) problem in a composition of $P$ under \pcc{} with a monitor that checks whether a \pcc{} behavior is an ``anomaly'', i.e., admitted by $P$ under \pcc{}, but not under \sic{}. 
This approach raises two non-trivial challenges: (1) defining a monitor for detecting \pcc{} vs \sic{} anomalies that uses a minimal amount of auxiliary memory (to remember past events), and (2) determining the complexity of checking if the composition of $P$ with the monitor reaches a specific control location\footnote{We assume that the monitor goes to an error location when detecting an anomaly.} under the (weaker) model \pcc{}. Interestingly enough, we address these two challenges by studying the relationship between these two weak consistency models, \pcc{} and \sic{}, and \emph{serializability}. The construction of the monitor is based on the fact that the \pcc{} vs \sic{} anomalies can be defined as roughly, the difference between the \pcc{} vs \serc{} and \sic{} vs \serc{} anomalies (investigated in previous work~\cite{DBLP:conf/concur/0002G16}), and we show that the reachability problem under \pcc{} can be reduced to a reachability problem under \serc{}. These results lead to a polynomial-time reduction of this robustness problem (for arbitrary programs) to a reachability problem under \serc{}, which is important from a practical point of view since the \serc{} semantics (as opposed to the \pcc{} or \sic{} semantics) can be encoded easily in existing verification tools (using locks to guard the isolation of transactions). These results also enable a precise characterization of the complexity class of this problem.

Checking robustness against substituting \pcc{} with \ccc{} is reduced to the problem of checking robustness against substituting \serc{} with \ccc{}. The latter has been shown to be polynomial-time reducible to reachability under \serc{} in~\cite{DBLP:conf/concur/BeillahiBE19}. This surprising result relies on the reduction from \pcc{} reachability to \serc{} reachability mentioned above. This reduction shows that a given program $P$ reaches a certain control location under \pcc{} iff a transformed program $P'$, where essentially, each transaction is split in two parts, one part containing all the reads, and one part containing all the writes, reaches the same control location under \serc{}. Since this reduction preserves the structure of the program, \ccc{} vs \pcc{} anomalies of a program $P$ correspond to \ccc{} vs \serc{} anomalies of the transformed program $P'$.

Beyond enabling these reductions, the characterization of classes of anomalies or the reduction from the \pcc{} semantics to the \serc{} semantics are also important for a better understanding of these weak consistency models and the differences between them. We believe that these results can find applications beyond robustness checking, e.g., verifying conformance to given specifications.


As a more pragmatic approach for establishing robustness, which avoids a non-reachability proof under \serc{}, we have introduced a proof methodology that builds on Lipton's reduction theory~\cite{DBLP:journals/cacm/Lipton75} and the concept of commutativity dependency graph introduced in~\cite{DBLP:conf/cav/BeillahiBE19}, which represents mover type dependencies between the transactions in a program. We give sufficient conditions for robustness in all the cases mentioned above, which characterize the commutativity dependency graph associated to a given program.

We tested the applicability of these verification techniques on a benchmark containing seven challenging applications extracted from previous work~\cite{DBLP:journals/pvldb/DifallahPCC13,DBLP:conf/cloud/HoltBZPOC16,DBLP:conf/pldi/BrutschyD0V18}. 
These techniques are precise enough for proving or disproving the robustness of all these applications, for all combinations of the consistency models. 


\tikzset{transactionBlock/.style={draw=black!0}}

\begin{figure}[!ht]
\lstset{basicstyle=\ttfamily\scriptsize,numbers=left,
keywords = {assume, select, return}, stepnumber=1,numberblanklines=false,mathescape=true}
\begin{subfigure}{70mm}
\begin{minipage}[t]{0.46\textwidth}
\begin{lstlisting}
      Process 1

CreateEvent(v, e1, 3):
[ Tickets[v][e1] := 3 ]

CountTickets(v):
[ r := $\sum\limits_{\mbox{e}}$Tickets[v][e] ]

\end{lstlisting}
\end{minipage}
\hfill
\begin{minipage}[t]{0.46\textwidth}
\begin{lstlisting}
      Process 2

CreateEvent(v, e2, 3):
[ Tickets[v][e2] := 3 ]

CountTickets(v):
[ r := $\sum\limits_{\mbox{e}}$Tickets[v][e] ]

\end{lstlisting}
\end{minipage}
\caption{FusionTicket.}
\label{fig:over-litmus0}
\end{subfigure}
\hfill
\begin{subfigure}{55mm}
    \scalebox{0.65}
    {
  \begin{tikzpicture}[->,>=stealth',shorten >=1pt,auto,node distance=4cm,
    semithick, transform shape,every text node part/.style={align=left}]
   \node[shape=rectangle ,draw=none,font=\large] at (-4,0)  (m)  {CreateEvent(v,e1,3)};
   \node[shape=rectangle ,draw=none,font=\large, label={right://r=3}] at (-4.3,-2)  (n)  {CountTickets(v)};
   \node[shape=rectangle ,draw=none,font=\large, ] at (0,0) (p){CreateEvent(v,e2,3)};
   \node[shape=rectangle ,draw=none,font=\large, label={right://r=3}] at (0,-2) (l) {CountTickets(v)};
   \begin{scope}[ every edge/.style={draw=black,very thick}]
   \path[->] (m.212) edge[left] node {$\hbo$} (n.119);
   \path[->] (m.230) edge[right,dashed] node {$\po$} (n.80);
   \path[->] (n) edge[below] node {$\hbo$} (p);
   \path[->] (p.240) edge[left] node {$\hbo$} (l.119);
   \path[->] (p.280) edge[right,dashed] node {$\po$} (l.80);
   \path[->] (l) edge[above] node {$\hbo$} (m);
  \end{scope}
  \end{tikzpicture}}
  \caption{A \ccc{} trace of FusionTicket.}
  \label{fig:over-litmus1}
  \end{subfigure}
  
\begin{subfigure}{73mm}
\begin{minipage}[t]{0.45\textwidth}
\begin{lstlisting}
      Process 1

Register(u, p1):
[ r := RegisteredUsers[u]
  assume r == 0
  RegisteredUsers[u] := 1
  Password[u] := p1 ]
\end{lstlisting}
\end{minipage}
\hfill
\begin{minipage}[t]{0.45\textwidth}
\begin{lstlisting}
      Process 2

Register(u, p2):
[ r := RegisteredUsers[u]
  assume r == 0
  RegisteredUsers[u] := 1
  Password[u] := p2 ]
\end{lstlisting}
\end{minipage}
\caption{Twitter.}
\label{fig:over-litmus2}
\end{subfigure}
\hfill
  \begin{subfigure}{45mm}
  \scalebox{0.65}
    {\begin{tikzpicture}[->,>=stealth',shorten >=1pt,auto,node distance=4cm,
    semithick, transform shape,every text node part/.style={align=left}]
   \node[shape=rectangle ,draw=none,font=\large] at (-3.5,0)  (m)  {Register(u,p1)};
   \node[shape=rectangle ,draw=none,font=\large] at (0,0) (p){Register(u,p2)};
   \begin{scope}[ every edge/.style={draw=black,very thick}]
   \path[->] (m) edge[bend right,above] node {$\hbo$} (p);
   \path[->] (p) edge[bend right,above] node {$\hbo$} (m);
  \end{scope}
  \end{tikzpicture}}
  \caption{A \ccc{} and \pcc{} trace of Twitter.}
  \label{fig:over-litmus3}
  \end{subfigure}

  \begin{subfigure}{73mm}
\begin{minipage}[t]{0.45\textwidth}
\begin{lstlisting}
      Process 1

RegisterRd(u, p1):
[ r := RegisteredUsers[u]
  assume r == 0 ]

RegisterWr(u, p1):
[ RegisteredUsers[u] := 1
  Password[u] := p1 ]
\end{lstlisting}
\end{minipage}
\hfill
\begin{minipage}[t]{0.45\textwidth}
\begin{lstlisting}
      Process 2

RegisterRd(u, p2):
[ r := RegisteredUsers[u]
  assume r == 0 ]

RegisterWr(u, p2):
[ RegisteredUsers[u] := 1
  Password[u] := p2 ]
\end{lstlisting}
\end{minipage}
\caption{Transformed Twitter.}
\label{fig:over-litmus4}
\end{subfigure}
\hfill
  \begin{subfigure}{50mm}
  \scalebox{0.65}
    {
  \begin{tikzpicture}[->,>=stealth',shorten >=1pt,auto,node distance=4cm,
    semithick, transform shape,every text node part/.style={align=left}]
    \node[shape=rectangle ,draw=none,font=\large] at (-4,0)  (m)  {RegisterRd(u,p1)};
    \node[shape=rectangle ,draw=none,font=\large] at (-4,-2)  (n)  {RegisterWr(u,p1)};
    \node[shape=rectangle ,draw=none,font=\large] at (0,0) (p){RegisterRd(u,p2)};
    \node[shape=rectangle ,draw=none,font=\large] at (0,-2) (l){RegisterWr(u,p2)};
    \begin{scope}[ every edge/.style={draw=black,very thick}]
   \path[->] (m) edge[above] node {$\hbo$} (l);
   \path[->] (m.235) edge[left] node {$\hbo$} (n.120);
   \path[->] (m.275) edge[right,dashed] node {$\po$} (n.80);
   \path[->] (n.360) edge[above] node {$\hbo$} (l.180);
   \path[->] (p) edge[below] node {$\hbo$} (n);
   \path[->] (p.235) edge[left] node {$\hbo$} (l.120);
   \path[->] (p.275) edge[right,dashed] node {$\po$} (l.80);
  \end{scope}
  \end{tikzpicture}}
  \caption{A \ccc{} and \serc{} trace of transformed Twitter.}
  \label{fig:over-litmus5}
  \end{subfigure}

\begin{subfigure}{130mm}
\begin{minipage}[t]{0.29\textwidth}
\begin{lstlisting}
      Process 1

PlaceBet(1,2):
[ assume time < TIMEOUT
    Bets[1] := 2 ]
\end{lstlisting}
\end{minipage}
\hfill
\begin{minipage}[t]{0.29\textwidth}
\begin{lstlisting}
      Process 2

PlaceBet(2,3):
[ assume time < TIMEOUT
    Bets[2] := 3 ]
\end{lstlisting}
\end{minipage}
\hfill
\begin{minipage}[t]{0.36\textwidth}
\begin{lstlisting}
      Process 3

SettleBet():
[Bets' := Bets
 n := Bets'.Length 
 assume time > TIMEOUT & n > 0
 select i s.t. Bets'[i] $\neq$ $\bot$
 return := Bets'[i] ]
\end{lstlisting}
\end{minipage}
\vspace{-5pt}
\caption{Betting.}
\label{fig:over-litmus6}
\end{subfigure}

\begin{subfigure}{75mm}
\scalebox{0.6}
{\begin{tikzpicture}[->,>=stealth',shorten >=1pt,auto,node distance=4cm,
  semithick, transform shape,every text node part/.style={align=left}]
 \node[shape=rectangle ,draw=none,font=\large] at (-6,0)  (m)  {PlaceBet(1,2)};
 \node[shape=rectangle ,draw=none,font=\large] at (-3,0)  (p)  {PlaceBet(2,3)};
 \node[shape=rectangle ,draw=none,font=\large, label={right:// return=2}] at (0.5,0) (q){SettleBet()};
 \begin{scope}[ every edge/.style={draw=black,very thick}]
 \path[->] (m) edge[bend right,above] node {$\hbo$} (q);
 \path[->] (q) edge[above] node {$\hbo$} (p);
\end{scope}
\end{tikzpicture}}
\caption{A \pcc{} and \sic{} trace of Betting.}
\label{fig:over-litmus7}
\end{subfigure}
\hfill
\begin{subfigure}{45mm}
  \scalebox{0.6}
  {\begin{tikzpicture}
   \node[shape=rectangle ,draw=none,font=\large] (A) at (0,0)  [] {PlaceBet(1,2)};
    \node[shape=rectangle ,draw=none,font=\large] (B) at (2.5,0)  [] {SettleBet()};
    \node[shape=rectangle ,draw=none,font=\large] (C) at (5,0)  [] {PlaceBet(2,3)};
    \begin{scope}[ every edge/.style={draw=black,very thick}]
    \path [->] (A) edge  [bend left] node [above,font=\large] {} (B);
    \path [->] (B) edge[bend left] node [above,font=\large] {} (A);
    \path [->] (C) edge  [bend right] node [above,font=\large] {} (B);
    \path [->] (B) edge[bend right] node [above,font=\large] {} (C);
    \end{scope}
  \end{tikzpicture}}
  \caption{Commutativity dependency graph of Betting.}
  \label{fig:over-litmus8}
\end{subfigure}
\vspace{-5pt}
\caption{Transactional programs and traces under different consistency models.}
\vspace{-16pt}
\label{fig:exmapleprograms}
\end{figure}
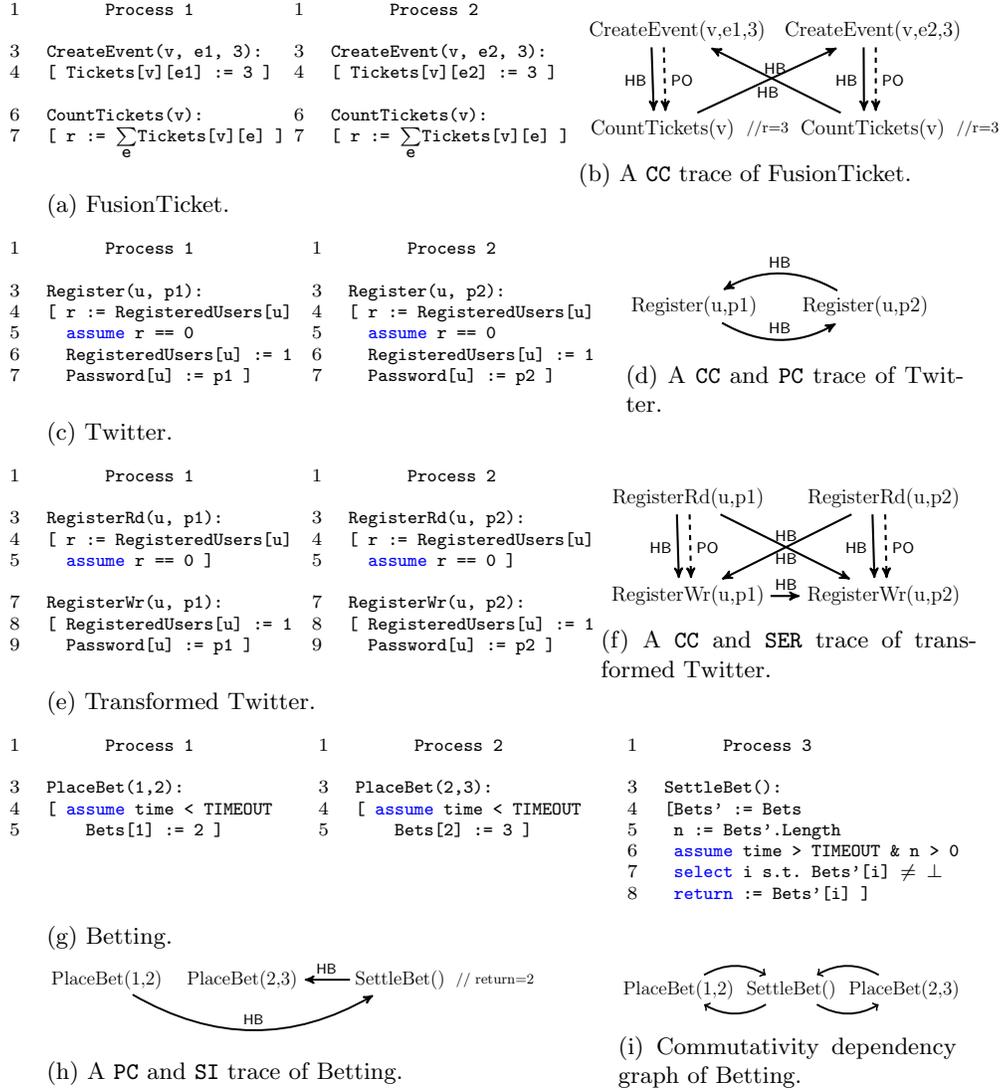

\vspace{-7pt}
\section{Overview}
\label{sec:overview}
\vspace{-2pt}
We give an overview of the robustness problems investigated in this paper, discussing first the case \pcc{} vs. \ccc{}, and then \sic{} vs \pcc{}. We end with an example that illustrates the robustness checking technique based on commutativity arguments.

\smallskip
\noindent
\textbf{Robustness \pcc{} vs \ccc{}.}
We illustrate the robustness against substituting \pcc{} with \ccc{} using the FusionTicket and the Twitter programs in Figure~\ref{fig:over-litmus0} and Figure~\ref{fig:over-litmus2}, respectively. FusionTicket manages tickets for a number of events, each event being associated with a venue. Its state consists of a two-dimensional map that stores the number of tickets for an event in a given venue ($r$ is a local variable, and the assignment in \texttt{CountTickets} is interpreted as a read of the shared state). The program has two processes and each process contains two transactions. The first transaction creates an event $\mbox{e}$ in a venue $\mbox{v}$ with a number of tickets $\mbox{n}$, and the second transaction computes the total number of tickets for all the events in a venue $\mbox{v}$. 
A possible candidate for a specification of this program is that the values computed in \texttt{CountTickets} are monotonically increasing since each such value is computed after creating a new event. Twitter provides a transaction for registering a new user with a given username and password, which is executed by two parallel processes. Its state contains two maps that record whether a given username has been registered (0 and 1 stand for non-registered and registered, respectively) and the password for a given username. Each transaction first checks whether a given username is free (see the \texttt{assume} statement).
The intended specification is that the user must be registered with the given password when the registration transaction succeeds.

A program is robust against substituting \pcc{} with \ccc{} if its set of behaviors under the two models coincide. 
We model behaviors of a given program as \emph{traces}, which record standard control-flow and data-flow dependencies between transactions, e.g., the order between transactions in the same session and whether a transaction reads the value written by another (read-from). 
The transitive closure of the union of all these dependency relations is called \emph{happens-before}. Figure \ref{fig:over-litmus1} pictures a trace of FusionTicket where the concrete values which are read in a transaction are written under comments. In this trace, each process registers a different event but in the same venue and with the same number of tickets, and it ignores the event created by the other process when computing the sum of tickets in the venue.

Figure \ref{fig:over-litmus1} pictures a trace of FusionTicket under \ccc{}, which is a witness that FusionTicket is \emph{not} robust against substituting \pcc{} with \ccc{}. This trace is also a violation of the intended specification since the number of tickets is not increasing (the sum of tickets is $3$ in both processes). The happens-before dependencies (pictured with $\hbo$ labeled edges) include the program-order $\po$ (the order between transactions in the same process), and read-write dependencies, since an instance of $\mbox{CountTickets(v)}$ does not observe the value written by the $\mbox{CreateEvent}$ transaction in the other process (the latter overwrites some value that the former reads).
This trace is allowed under \ccc{} because the transaction 
$\mbox{CreateEvent(v, e1, 3)}$ executes concurrently with the transaction $\mbox{CountTickets(v)}$ in the other process, and similarly for $\mbox{CreateEvent(v, e2, 3)}$. However, it is not allowed under \pcc{} since it is impossible to define a total commit order between $\mbox{CreateEvent(v, e1, 3)}$ and $\mbox{CreateEvent(v, e2, 3)}$ that justifies the reads of both $\mbox{CountTickets(v)}$ transactions (these reads should correspond to the updates in a prefix of this order). For instance, assuming that $\mbox{CreateEvent(v, e1, 3)}$ commits before $\mbox{CreateEvent(v, e2, 3)}$, $\mbox{CountTickets(v)}$ in the second process must observe the effect of $\mbox{CreateEvent(v, e1, 3)}$ as well since it observes the effect of $\mbox{CreateEvent(v, e2, 3)}$. However, this contradicts the fact that $\mbox{CountTickets(v)}$ computes the sum of tickets as being $3$. 

On the other hand, Twitter is robust against substituting \pcc{} with \ccc{}. For instance, Figure \ref{fig:over-litmus3} pictures a trace of Twitter under \ccc{}, where the \texttt{assume} in both transactions pass. In this trace, the transactions $\mbox{Register(u,p1)}$ and $\mbox{Register(u,p2)}$ execute concurrently and are unaware of each other's writes (they are not causally related). 
The $\hbo$ dependencies include write-write dependencies since both transactions write on the same location (we consider the transaction in Process 2 to be the last one writing to the \texttt{Password} map), and read-write dependencies since each transaction reads  \texttt{RegisteredUsers} that is written by the other.
This trace is also allowed under \pcc{} since the commit order can be defined such that $\mbox{Register(u,p1)}$ is ordered before $\mbox{Register(u,p2)}$, and then both transactions read from the initial state (the empty prefix). Note that this trace has a cyclic happens-before which means that it is not allowed under serializability. 

\smallskip
\noindent
\textbf{Checking robustness \pcc{} vs \ccc{}.}
We reduce the problem of checking robustness against substituting \pcc{} with \ccc{} to the robustness problem against substituting \serc{} with \ccc{} (the latter reduces to a reachability problem under \serc{}~\cite{DBLP:conf/concur/BeillahiBE19}). This reduction relies on a syntactic program transformation that rewrites \pcc{} behaviors of a given program $P$ to $\serc{}$ behaviors of another program $P'$. The program $P'$ is obtained 
by splitting each transaction $\atr$ of $P$ into two transactions: the first transaction performs all the reads in $\atr$ and the second performs all the writes in $\atr$ (the two are related by program order). Figure \ref{fig:over-litmus4} shows this transformation applied on Twitter. The trace in Figure \ref{fig:over-litmus5} is a serializable execution of the transformed Twitter which is ``observationally'' equivalent to the trace in Figure~\ref{fig:over-litmus3} of the original Twitter, i.e., each read of the shared state returns the same value and the writes on the shared state are applied in the same order (the acyclicity of the happens-before shows that this is a serializable trace). The transformed FusionTicket  coincides with the original version because it contains no transaction that both reads and writes on the shared state. 

We show that \pcc{} behaviors and \serc{} behaviors of the original and transformed program, respectively, are related by a bijection. In particular, we show that any \pcc{} vs. \ccc{} robustness violation of the original program manifests as a \serc{} vs. \ccc{} robustness violation of the transformed program, and vice-versa. For instance, the \ccc{} trace of the original Twitter in Figure~\ref{fig:over-litmus3} corresponds to the \ccc{} trace of the transformed Twitter in Figure \ref{fig:over-litmus5}, and the acyclicity of the latter (the fact that it is admitted by \serc{}) implies that the former is admitted by the original Twitter under \pcc{}. On the other hand, the trace in Figure \ref{fig:over-litmus1} is also a \ccc{} of the transformed FusionTicket and its cyclicity implies that it is not admitted by FusionTicket under \pcc{}, and thus, it represents a robustness violation.


\smallskip
\noindent
\textbf{Robustness \sic{} vs \pcc{}.}
We illustrate the robustness against substituting \sic{} with \pcc{} using Twitter and the Betting program in Figure \ref{fig:over-litmus6}. Twitter is \emph{not} robust against substituting \sic{} with \pcc{}, the trace in Figure \ref{fig:over-litmus3} being a witness violation. This trace is also a violation of the intended specification since one of the users registers a password that is overwritten in a concurrent transaction. 
This \pcc{} trace is not possible under \sic{} because $\mbox{Register(u,p1)}$ and $\mbox{Register(u,p2)}$ observe the same prefix of the commit order (i.e., an empty prefix), but they write to a common memory location $\mbox{Password[u]}$ which is not allowed under \sic{}. 

On the other hand, the Betting program in Figure \ref{fig:over-litmus6}, which manages a set of bets, is robust against substituting \sic{} with \pcc{}. The first two processes execute one transaction that places a bet of a value $\mbox{v}$ with a unique bet identifier $\mbox{id}$, assuming that the bet expiration time is not yet reached (bets are recorded in the map \texttt{Bets}). The third process contains a single transaction that settles the betting assuming that the bet expiration time was reached and at least one bet has been placed. This transaction starts by taking a snapshot of the \texttt{Bets} map into a local variable \texttt{Bets'}, and then selects a random non-null value (different from $\bot$) in the map to correspond to the winning bet. 
The intended specification of this program is that the winning bet corresponds to a genuine bet that was placed. 
Figure \ref{fig:over-litmus6} pictures a \pcc{} trace of Betting where $\mbox{SettleBet}$ observes only the bet of the first process $\mbox{PlaceBet(1,2)}$. The $\hbo$ dependency towards the second process denotes a read-write dependency ($\mbox{SettleBet}$ reads a cell of the map \texttt{Bets} which is overwritten by the second process). This trace is allowed under \sic{} because no two transactions write to the same location. 

\smallskip
\noindent
\textbf{Checking robustness \sic{} vs \pcc{}.}
We reduce robustness against substituting \pcc{} with \ccc{} to a reachability problem under \serc{}. This reduction is based on a characterization of happens-before cycles\footnote{Traces with an acyclic happens-before are not robustness violations because they are admitted under serializability, which implies that they are admitted under the weaker model \sic{} as well.} that are possible under \pcc{} but not \sic{}, and the transformation described above that allows to simulate the \pcc{} semantics of a program on top of \serc{}. The former is used to define an instrumentation (monitor) for the transformed program that reaches an error state iff the original program is not robust. 
Therefore, we show that the happens-before cycles in \pcc{} traces that are not admitted by \sic{} must contain a transaction that (1) overwrites a value written by another transaction in the cycle and (2) reads a value overwritten by another transaction in the cycle. 
For instance, the trace of Twitter in Figure \ref{fig:over-litmus3} is not allowed under \sic{} because $\mbox{Register(u,p2)}$ overwrites a value written by $\mbox{Register(u,p1)}$ (the password) and reads a value overwritten by $\mbox{Register(u,p1)}$ (checking whether the username $u$ is registered). 
The trace of Betting in Figure \ref{fig:over-litmus6} is allowed under \sic{} because its happens-before is acyclic.

\smallskip
\noindent
\textbf{Checking robustness using commutativity arguments.}
Based on the reductions above, we propose an approximated method for proving robustness based on the concept of mover in Lipton's reduction theory~\cite{DBLP:journals/cacm/Lipton75}. A transaction is a left (resp., right) mover if it commutes to the left (resp., right) of another transaction (by a different process) while preserving the computation. We use the notion of mover to characterize the data-flow dependencies in the happens-before. Roughly, there exists a data-flow dependency between two transactions in some execution if one doesn't commute to the left/right of the other one.

We define a commutativity dependency graph which summarizes the happens-before dependencies in all executions of a transformed program (obtained by splitting the transactions of the original program as explained above), and derive a proof method for robustness which inspects paths in this graph. Two transactions $\atr_1$ and $\atr_2$ are linked by a directed edge iff $\atr_1$ \emph{cannot} move to the right of $\atr_2$ (or $\atr_2$ cannot move to the left of $\atr_1$), or if they are related by the program order. Moreover, two transactions $\atr_1$ and $\atr_2$ are linked by an undirected edge iff they are the result of splitting the same transaction. 

A program is robust against substituting \pcc{} with \ccc{} if roughly, its commutativity dependency graph does \emph{not} contain a \emph{simple} cycle of directed edges with two distinct transactions $\atr_1$ and $\atr_2$, such that $\atr_1$ does not commute left because of another transaction $\atr_3$ in the cycle that reads a variable that $\atr_1$ writes to, and $\atr_2$ does not commute right because of another transaction $\atr_4$ in the cycle ($\atr_3$ and $\atr_4$ can coincide) that writes to a variable that $\atr_2$ either reads from or writes to\footnote{The transactions $\atr_1$, $\atr_2$, $\atr_3$, and $\atr_4$ correspond to $\atr_1$, $\atr_i$, $\atr_n$, and $\atr_{i+1}$, respectively, in Theorem~\ref{them:MovRobCcPc}.}. 
For instance, Figure \ref{fig:over-litmus8} shows the commutativity dependency graph of the transformed Betting program, which coincides with the original Betting because $\mbox{PlaceBet(1,2)}$ and $\mbox{PlaceBet(2,3)}$ are write-only transactions and $\mbox{SettleBet()}$ is a read-only transaction. Both simple cycles in Figure \ref{fig:over-litmus8} contain just two transactions and therefore do not meet the criterion above which requires at least 3 transactions. Therefore, Betting is robust against substituting \pcc{} with \ccc{}. 

A program is robust against substituting \sic{} with \pcc{}, if roughly, its commutativity dependency graph does \emph{not} contain a \emph{simple} cycle with two successive transactions $\atr_1$ and $\atr_2$ that are linked by an undirected edge, such that $\atr_1$ does not commute left because of another transaction $\atr_3$ in the cycle that writes to a variable that $\atr_1$ writes to, and $\atr_2$ does not commute right because of another transaction $\atr_4$ in the cycle ($\atr_3$ and $\atr_4$ can coincide) that writes to a variable that $\atr_2$ reads from\footnote{The transactions $\atr_1$, $\atr_2$, $\atr_3$, and $\atr_4$ correspond to $\atr_1$, $\atr_2$, $\atr_n$, and $\atr_3$, respectively, in Theorem~\ref{them:MovRobPcSi}.}. Betting is also robust against substituting \sic{} with \pcc{} for the same reason (simple cycles of size 2).
\vspace{-10pt}
\section{Consistency Models}
\vspace{-5pt}
\label{sec:consistency}

\begin{figure}[t]
  {\footnotesize
  \setlength{\grammarindent}{7em}
  \begin{grammar}
  <prog>  ::= \plog{program} <process>$^{*}$
  
  <process> ::= \plog{process} <pid> \plog{regs} <reg>$^{*}$ <txn>$^{*}$
  
  <txn> ::=  \plog{begin} <read>$^{*}$ <test>$^{*}$ <write>$^{*}$ \plog{commit}
  
  <read> ::= <label>":" <reg> ":=" <var>";" \plog{goto} <label>";"

  <test> ::= <label>":" \plog{assume} <bexpr>";" \plog{goto} <label>";"

  <write> ::= <label>":" <var> ":=" <reg-expr>";" \plog{goto} <label>";"
  \end{grammar}}
  \vspace{-10pt}
  \caption{The syntax of our programming language. $a^{*}$ indicates zero or more occurrences of $a$.  $\langle pid\rangle$, $\langle reg\rangle$, $\langle label \rangle$, and $\langle var\rangle$ represent a process identifier, a register, a label, and a shared variable, respectively. $\langle reg\text{-}expr \rangle$ is an expression over registers while $\langle bexpr \rangle$ is a Boolean expression over registers, or the non-deterministic choice $*$.
  }
  \label{Figure:syntax}
  \vspace{-13pt}
\end{figure}

\noindent
{\bf Syntax.}
We present our results in the context of the simple programming language, defined in Figure~\ref{Figure:syntax}, where a program is a parallel composition of \emph{processes} distinguished using a set of identifiers $\mathbb{P}$. 
A process is a sequence of \emph{transactions} and each transaction is a sequence of \emph{labeled instructions}.
A transaction starts with a \plog{begin} instruction and finishes with a \plog{commit} instruction. 
Instructions include assignments to a process-local \emph{register} from a set $\mathbb{R}$ or to a \emph{shared variable} from a set $\mathbb{V}$, or an \plog{assume}.
The assignments use values from a data domain $\mathbb{D}$.
An assignment to a register $\langle reg\rangle := \langle var\rangle$ is called a \emph{read} of the shared-variable $\langle var\rangle$ and an assignment to a shared variable $\langle var\rangle := \langle reg\rangle$ is called a \emph{write} to the shared-variable $\langle var\rangle$.
The \plog{assume} $\langle bexpr\rangle$ blocks the process if the Boolean expression $\langle bexpr\rangle$ over registers is false. It can be used to model conditionals. The \plog{goto} statement transfers the control to the program location (instruction) specified by a given label. Since multiple instructions can have the same label, \plog{goto} statements can be used to mimic imperative constructs like loops and conditionals inside transactions. 


We assume w.l.o.g. that every transaction is written as a sequence of reads or {\tt assume} statements followed by a sequence of writes (a single {\tt goto} statement from the sequence of read/{\tt assume} instructions transfers the control to the sequence of writes). In the context of the consistency models we study in this paper, every program can be equivalently rewritten as a set of transactions of this form. 

To simplify the technical exposition, programs contain a bounded number of processes and each process executes a bounded number of transactions. A transaction may execute an unbounded number of instructions but these instructions concern a bounded number of variables, which makes it impossible to model SQL (select/update) queries that may access tables with a statically unknown number of rows. Our results can be extended beyond these restrictions as explained in Remark~\ref{rem:robustness} and Remark~\ref{rem:comgraph}.


\noindent
{\bf Semantics.}
We describe the semantics of a program under four consistency models, i.e., causal consistency\footnote{We consider a variation known as causal convergence~\cite{DBLP:journals/ftpl/Burckhardt14,DBLP:conf/popl/BouajjaniEGH17}} (\ccc{}), prefix consistency (\pcc{}), snapshot isolation (\sic{}), and serializability (\serc{}).



In the semantics of a program under \ccc{}, shared variables are replicated across each process, each process maintaining its own local valuation of these variables. During the execution of a transaction in a process, its writes are stored in a \emph{transaction log} that can be accessed only by the process executing the transaction and that is broadcasted to all the other processes at the end of the transaction. To read a shared variable $\anaddr$, a process $\apr$ first accesses its transaction log and takes the last written value on $\anaddr$, if any, and then its own valuation of the shared variable, if  $\anaddr$ was not written during the current transaction. Transaction logs are delivered to every process in an order consistent with the \emph{causal} relation between transactions, i.e., the transitive closure of the union of the \emph{program order} (the order in which transactions are executed by a process), and the \emph{read-from} relation (a transaction $\atr_1$ reads-from a transaction $\atr_2$ iff $\atr_1$ reads a value that was written by $\atr_2$). When a process receives a transaction log, it immediately applies it on its shared-variable valuation. 

In the semantics of a program under \pcc{} and \sic{}, shared variables are stored in a central memory and each process keeps a local valuation of these variables. When a process starts a new transaction, it fetches a consistent snapshot of the shared variables from the central memory and stores it in its local valuation of these variables. 
During the execution of a transaction in a process, writes to shared variables are stored in the local valuation  of these variables, and in a transaction log. To read a shared variable, a process takes its own valuation of the shared variable. A process commits a transaction by applying the updates in the transaction log on the central memory in an atomic way (to make them visible to all processes). Under \sic{}, when a process applies the writes in a transaction log on the central memory, it must ensure that there were no concurrent writes that occurred after the last fetch from the central memory to a shared variable that was written during the current transaction. Otherwise, the transaction is aborted and its effects discarded.

In the semantics of a program under \serc{}, we adopt a simple operational model where we keep a single shared-variable valuation in a central memory (accessed by all processes) with the standard interpretation of read and write statements. Transactions execute serially, one after another.


We use a standard model of executions of a program called \emph{trace}. A trace represents the order between transactions in the same process, and the data-flow in an execution using standard happens-before relations between transactions. We assume that each transaction in a program is identified uniquely using a transaction identifier from a set $\mathbb{T}$. 
Also, $\amap: \mathbb{T} \rightarrow 2^{\mathbb{S}}$ is a mapping that associates each transaction in $\mathbb{T}$ with a sequence of read and write events from the set 
\begin{align*}
  \mathbb{S} = \{\readact(\atr,\anaddr,\aval), \writeact(\atr,\anaddr,\aval): \atr\in \mathbb{T}, \anaddr\in \mathbb{V}, \aval\in \mathbb{D}\}
\end{align*}
  where $\readact(\atr,\anaddr,\aval)$ is a read of $\anaddr$ returning $\aval$, and $\writeact(\atr,\anaddr,\aval)$ is a write of $\aval$ to $\anaddr$.

\begin{definition}
  A \emph{trace} is a tuple $\atrace = (\rho,\amap,\tor,\po,\rfo,\sto,\cfo)$ where $\rho\subseteq \mathbb{T}$ is a set of transaction identifiers, and
\begin{itemize}
\item $\tor$ is a mapping giving the order between events in each transaction, i.e., it associates each transaction $\atr$ in $\rho$ with a total order $\tor(\atr)$ on $\amap(\atr) \times \amap(\atr)$.
\item  $\po$ is the program order relation, a strict partial order on $\rho \times \rho$ that orders every two transactions issued by the same process.
\item  $\rfo$ is the read-from relation between distinct transactions $(\atr1, \atr2) \in \rho \times \rho$ 
representing the fact that $\atr2$ reads a value written by $\atr1$.
\item  $\sto$ is the store order relation on $\rho \times \rho$ between distinct transactions that write to the same shared variable. 
\item $\cfo$ is the conflict order relation between distinct transactions, defined by $\cfo = \rfo^{-1};\sto$ ($;$ denotes the sequential composition of two relations).
\end{itemize}
\end{definition}

For simplicity, for a trace $\atrace = (\rho,\amap,\tor,\po,\rfo,\sto,\cfo)$, we write $t\in \atrace$ instead of $t\in\rho$.
We also assume that each trace contains a fictitious transaction that writes the initial values of all shared variables, and which is ordered before any other transaction in program order.
Also, $\tracesconf_{\textsf{X}}(\aprog)$ is the set of traces representing executions of program $\aprog$ under a consistency model $\textsf{X}$. 

For each $\textsf{X}\in \{\ccc{},\pcc{},\sic{},\serc{}\}$, the set of traces $\tracesconf_{\textsf{X}}(\aprog)$ can be described using the set of properties in Table \ref{fig:defaxioms}.
%
A trace $\atrace$ is possible under causal consistency iff there exist two relations $\viso$ a partial order (causal order) and 
$\arbo$ a total order (arbitration order) that includes $\viso$, such that the properties $\axpoco$, $\axcoarb$, and $\axretval$ hold \cite{DBLP:conf/concur/CeroneGY17,DBLP:conf/popl/BouajjaniEGH17}. 
$\axpoco$ guarantees that the program order and the read-from relation are included in the causal order, and 
$\axcoarb$ guarantees that the causal order and the store order are included in the arbitration order. 
$\axretval$ guarantees that a read returns the value written by the last write in the last transaction that contains a write to the same variable and that is ordered by $\viso$ before the read's transaction. 
We use $\axcc$ to denote the conjunction of these three properties.
A trace $\atrace$ is possible under prefix consistency iff there exist a causal order $\viso$ and an arbitration order
$\arbo$ such that $\axcc$ holds and the property $\axprefix$ holds as well \cite{DBLP:conf/concur/CeroneGY17}. 
$\axprefix$ guarantees that every transaction observes a prefix of transactions that are ordered by $\arbo$ before it. 
We use $\axpc$ to denote the conjunction of $\axcc$ and $\axprefix$. 
A trace $\atrace$ is possible under snapshot isolation iff there exist a causal order $\viso$ and an arbitration order
$\arbo$ such that $\axpc$ holds and the property $\axconflict$ holds \cite{DBLP:conf/concur/CeroneGY17}. 
$\axconflict$ guarantees that if two transactions write to the same variable then one of them must observe the other. 
We use $\axsi$ to  denote the conjunction of $\axpc$ and $\axconflict$. 
A trace $\atrace$ is serializable iff there exist a causal order $\viso$ and an arbitration order
$\arbo$ such that the property $\axser$ holds which implies that the two relations $\viso$ and $\arbo$ coincide. 
Note that for any given program $\aprog$, $\tracesconf_{\serc{}}(\aprog)\subseteq \tracesconf_{\sic{}}(\aprog)\subseteq \tracesconf_{\pcc{}}(\aprog)\subseteq \tracesconf_{\ccc{}}(\aprog)$. Also, the four consistency models we consider disallow anomalies such as dirty and phantom reads. 

\begin{table}[t]

    \begin{tabular}{l|l}
    $\axpoco$ & $\viso_{0}^{+} \subseteq \viso$ \\
     $\axcoarb$ & $\arbo_{0}^+ \subseteq \arbo$\\
     $\axcc$ & $\axretval \wedge \axpoco \wedge \axcoarb$  \\
     $\axprefix$ & $\arbo ; \viso \subseteq \viso$ \\
     $\axpc$ &  $ \axprefix \wedge \axcc$ \\
     $\axconflict$ & $\sto \subseteq \viso$\\
     $\axsi$ & $\axconflict \wedge \axpc$ \\
     $\axser$ & $\axretval \wedge \axpoco \wedge \axcoarb \wedge \viso = \arbo  $  \\
    \end{tabular}
    
    \vspace{1mm}
    where 

    $\viso_{0} = \po \cup \rfo$ and 
    $\arbo_{0} = \po \cup \rfo \cup \sto$

$\axretval$ = $\forall\ t\in \atrace.\ \forall\ \readact(\atr,\anaddr,\aval) \in \amap(\atr)$ we have that
    \begin{itemize}
        \item there exist a transaction $\atr_0 = Max_{\arbo}(\{\atr' \in \atrace\ |\ (\atr',\atr) \in \viso \wedge \exists\ \writeact(\atr',\anaddr,\cdot)\in\amap(\atr')\})$ and an event $\writeact(\atr_0,\anaddr,\aval) = Max_{\tor(\atr_0)}(\{\writeact(\atr_0,\anaddr,\cdot) \in \amap(\atr_0)\})$.
    \end{itemize}
    \caption{Declarative definitions of consistency models. For an order relation $\leq$, $a = Max_{\leq}(A)$ iff $a \in A \wedge \forall\ b \in A.\ b \leq a$.}
    \label{fig:defaxioms}
    \vspace{-25pt}
\end{table}
%

%
For a given trace $\atrace=(\rho,\amap,\tor, \po, \rfo, \sto, \cfo)$, the happens before order is the transitive closure of the union of all the relations in the trace, i.e., $\hbo = (\po \cup \rfo \cup \sto \cup \cfo)^{+}$. 
A classic result states that a trace $\atrace$ is serializable iff $\hbo$ is acyclic~\cite{Adya99,DBLP:journals/toplas/ShashaS88}. 
Note that $\hbo$ is acyclic implies that $\sto$ is a total order between transactions that write to the same variable, and 
$(\po \cup \rfo)^{+}$ and $(\po \cup \rfo \cup \sto)^{+}$ are acyclic. 



\vspace{-10pt}
\subsection{Robustness}\label{subsec:rob}

In this work, we investigate the problem of checking whether a program $\aprog$ under a semantics $\textsf{Y} \in \{\pcc{},\ \sic{}\}$ produces the same set of traces as under a weaker semantics $\textsf{X} \in \{\ccc{},\ \pcc{}\}$. When this holds, we say that $\aprog$ is \emph{robust} against $\textsf{X}$ relative to $\textsf{Y}$. 

\begin{definition}
A program $\aprog$ is called \emph{robust} against a semantics $\textsf{X} \in \{\ccc{},\ \pcc{},\ \sic{}\}$ relative to a semantics $\textsf{Y} \in \{\pcc{},\ \sic{},\ \serc{}\}$ such that $\textsf{Y}$ is stronger than $\textsf{X}$ iff $\tracesconf_{\textsf{X}}(\aprog)=\tracesconf_{\textsf{Y}}(\aprog)$.
\end{definition}

If $\aprog$ is not robust against $\textsf{X}$ relative to $\textsf{Y}$ then there must exist a trace $\atrace \in \tracesconf_{\textsf{X}}(\aprog) \setminus \tracesconf_{\textsf{Y}}(\aprog)$. 
We say that $\atrace$ is a robustness violation trace. 

\tikzset{transactionBlock/.style={draw=black!0}}
\begin{wrapfigure}{r}{0.39\textwidth}
\vspace{-20pt}
\lstset{basicstyle=\ttfamily\scriptsize}
\begin{subfigure}{55mm}
\scalebox{0.71}
{
\begin{tikzpicture}[->,>=stealth',shorten >=1pt,auto,node distance=4cm,
  semithick, transform shape,every text node part/.style={align=left}]
 \node[shape=rectangle ,draw=none,font=\large, label={left:$\atr_1$}] at (-2.5,0)  (m)  {$[x := 1]$};
 \node[shape=rectangle ,draw=none,font=\large, label={left:$\atr_2$}] at (-2.5,-2)       (n)           {$[r1 := y]\ //0$};
 \node[shape=rectangle ,draw=none,font=\large, label={right:$\atr_3$}] at (0,0) (p){$[y := 1]$};
 \node[shape=rectangle ,draw=none,font=\large, label={right:$\atr_4$}] at (0,-2) (l) {$[r2 := x]\ //0$};
 \begin{scope}[ every edge/.style={draw=black,very thick}]
 \path[->] (m) edge[left] node {$\po$} (n);
 \path[->] (n) edge[left] node {$\cfo$} (p);
 \path[->] (p) edge[right] node {$\po$} (l);
 \path[->] (l) edge[right] node {$\cfo$} (m);
\end{scope}
\end{tikzpicture}}
\caption{Store Buffering ($\mathsf{SB}$).}
\label{fig:litmus1}
\end{subfigure}
\hspace{2mm}
\begin{subfigure}{75mm}
\scalebox{0.71}
{
\begin{tikzpicture}[->,>=stealth',shorten >=1pt,auto,node distance=4cm,
  semithick, transform shape,every text node part/.style={align=left}]
 \node[shape=rectangle ,draw=none,font=\large, label={left:$\atr_1$}] at (-3,0)  (m)  {$[r1 := x\ \ //0$ \\ $\ x := r1 + 1]$};
 \node[shape=rectangle ,draw=none,font=\large, label={right:$\atr_2$}] at (0,0) (p){$[r2 := x\ \  //0$ \\ $\ x := r2 + 1]$};
 \begin{scope}[ every edge/.style={draw=black,very thick}]
 \path[->] (m) edge[bend right,above] node {$\sto$} (p);
 \path[->] (p) edge[bend right,above] node {$\cfo$} (m);
\end{scope}
\end{tikzpicture}}
\caption{Lost Update ($\mathsf{LU}$).}
\label{fig:litmus2}
\end{subfigure}
    
\begin{subfigure}{65mm}
\scalebox{0.71}
{
\begin{tikzpicture}[->,>=stealth',shorten >=1pt,auto,node distance=4cm,
  semithick, transform shape,every text node part/.style={align=left}]
 \node[shape=rectangle ,draw=none,font=\large, label={left:$\atr_1$}] at (-3,0)  (m)  {$[r1 := x\ \ //0$ \\ $\ y := 1]$};
 \node[shape=rectangle ,draw=none,font=\large, label={right:$\atr_2$}] at (0,0) (p){$[r2 := y\ \ //0$ \\ $\ x := 1]$};
 \begin{scope}[ every edge/.style={draw=black,very thick}]
 \path[->] (m) edge[bend right,above] node {$\cfo$} (p);
 \path[->] (p) edge[bend right,above] node {$\cfo$} (m);
\end{scope}
\end{tikzpicture}}
\caption{Write Skew (WS).}
\label{fig:litmus3}
\end{subfigure}
\hfill
\begin{subfigure}{60mm}
  \scalebox{0.71}
  {
\begin{tikzpicture}[->,>=stealth',shorten >=1pt,auto,node distance=4cm,
 semithick, transform shape,every text node part/.style={align=left}]
\node[shape=rectangle ,draw=none,font=\large, label={left:$\atr_1$}] at (-2.5,0)  (m)  {$[x := 1]$};
\node[shape=rectangle ,draw=none,font=\large, label={left:$\atr_2$}] at (-2.5,-2) (n)  {$[y := 1]$};
\node[shape=rectangle ,draw=none,font=\large, label={right:$\atr_3$}] at (0,0) (p){$[r1 := y]\ //1$};
\node[shape=rectangle ,draw=none,font=\large, label={right:$\atr_4$}] at (0,-2) (l) {$[r2 := x]\ //1$};
\begin{scope}[ every edge/.style={draw=black,very thick}]
\path[->] (m) edge[left] node {$\po$} (n);
\path[->] (n) edge[left] node {$\rfo$} (p);
\path[->] (p) edge[right] node {$\po$} (l);
\path[->] (m) edge[right] node {$\rfo$} (l);
\end{scope}
\end{tikzpicture}}
\caption{Message Passing (MP).}
\label{fig:litmus4}
\end{subfigure}
\vspace{-5pt}
\caption{Litmus programs}\label{fig:litmustests}
\vspace{-20pt}
\end{wrapfigure}
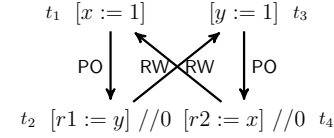
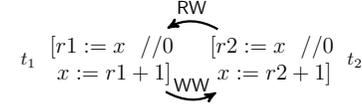
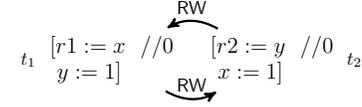
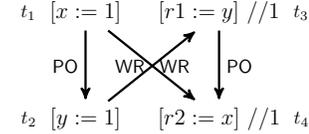

We illustrate the notion of robustness on the programs in Figure \ref{fig:litmustests}, which are commonly used in the literature. 
In all programs, transactions of the same process are aligned vertically and ordered from top to bottom.
Each read instruction is commented with the value it reads in some execution.


The store buffering ($\mathsf{SB}$) program in Figure \ref{fig:litmus1} contains four transactions that are issued by two distinct processes. 
We emphasize an execution where $\atr_2$ reads $0$ from $y$ and $\atr_4$ reads $0$ from $x$. 
This execution is allowed under \ccc{} since the two writes by $\atr_1$ and $\atr_3$ are not causally dependent. 
Thus, $\atr_2$ and $\atr_4$ are executed without seeing the writes from $\atr_3$ and $\atr_1$, respectively. 
However, this execution is not feasible under \pcc{} (which implies that it is not feasible under both \sic{} and \serc{}). 
In particular, we cannot have neither $(\atr_1,\atr_3) \in \arbo$ nor $(\atr_3,\atr_1) \in \arbo$ which contradicts the fact that $\arbo$ is total order. 
For example, if $(\atr_1,\atr_3) \in \arbo$, then $(\atr_1,\atr_4) \in \viso$ (since $\arbo;\viso \subset  \viso$) which contradicts the fact 
that $\atr_4$ does not see $\atr_1$. Similarly, $(\atr_3,\atr_1) \in \arbo$ implies that $(\atr_3,\atr_2) \in \viso$ which contradicts the fact 
that $\atr_2$ does not see $\atr_3$. Thus, $\mathsf{SB}$ is not robust against \ccc{} relative to \pcc{}.

The lost update ($\mathsf{LU}$) program in Figure \ref{fig:litmus2} has two transactions that are issued by two distinct processes. 
We highlight an execution where both transactions read $0$ from $x$. 
This execution is allowed under \pcc{} since both transactions are not causally dependent and can be executed in parallel by the two processes. However, it is not allowed under \sic{} since both transactions write to a common variable (i.e., $x$). 
Thus, they cannot be executed in parallel and one of them must see the write of the other. Thus, $\mathsf{SB}$ is not robust against \pcc{} relative to \sic{}.

The write skew ($\mathsf{WS}$) program in Figure \ref{fig:litmus3} has two transactions that are issued by two distinct processes. We highlight an execution where $\atr_1$ reads $0$ from $x$ and $\atr_2$ reads $0$ from $y$. This execution is allowed under \sic{} since both transactions are not causally dependent, do not write to a common variable, and can be executed in parallel by the two processes. However, this execution is not allowed under \serc{} since one of the two transactions must see the write of the other. 
Thus, $\mathsf{WS}$ is not robust against \sic{} relative to \serc{}. 

The message passing  ($\mathsf{MP}$) program in Figure \ref{fig:litmus4} has four transactions issued by two processes. 
Because $\atr_1$ and $\atr_2$ are causally dependent, under any semantics $\textsf{X} \in \{\ccc{},\ \pcc{},\ \sic{},\ \serc{}\}$ we only have three possible executions of $\mathsf{MP}$,
which correspond to either $\atr_3$ and $\atr_4$ not observing the writes of $\atr_1$ and $\atr_2$, or $\atr_3$ and $\atr_4$ observe the writes of both $\atr_1$ and $\atr_2$, or $\atr_4$ observes the write of $\atr_1$ (we highlight the values read in the second case in Figure \ref{fig:litmus4}). 
Therefore, the executions of this program under the four consistency models coincide. Thus, $\mathsf{MP}$ is robust against \ccc{} relative to any other model.

\vspace{-10pt}
\section{Robustness Against \ccc{} Relative to \pcc{}}
\label{sec:CCPCrobustness}
\vspace{-5pt}
We show that checking robustness against \ccc{} relative to \pcc{} can be reduced to checking robustness against \ccc{} relative to \serc{}. 
The crux of this reduction is a program transformation that allows to simulate the \pcc{} semantics of a program $\aprog$ using the \serc{} semantics of a program $\aprog_\pcinstr$. 
Checking robustness against \ccc{} relative to \serc{} can be reduced in polynomial time to reachability under \serc{}~\cite{DBLP:conf/concur/BeillahiBE19}.

Given a program $\aprog$ with a set of transactions $\trsaprog{\aprog}$, we define a program $\aprog_\pcinstr$ 
such 
that every transaction $\atr\in \trsaprog{\aprog}$ is split into a transaction $\atrrd{\atr}$ that contains all the read/{\tt assume} statements in $\atr$ (in the same order) and 
another transaction $\atrwr{\atr}$ that contains all the write statements in $\atr$ (in the same order). 
In the following, we establish the following result:

\begin{theorem} \label{them:RobCcPc}
A program $\aprog$ is robust against \ccc{} relative to \pcc{} iff  $\aprog_\pcinstr$ is robust against \ccc{} relative to \serc{}.
\end{theorem}

Intuitively, under \pcc{}, processes can execute concurrent transactions that fetch the same consistent snapshot of the shared variables from the central memory and subsequently commit their writes. Decoupling the read part of a transaction from the write part allows to simulate such behaviors even under \serc{}.


The proof of this theorem relies on several intermediate results concerning the relationship between traces of $\aprog$ and $\aprog_\pcinstr$. 
Let $\atrace= (\rho, \po, \rfo, \sto, \cfo) \in \tracesconf_{\textsf{X}}(\aprog)$ be a trace of a program $\aprog$ under a semantics $\textsf{X}$. We define the trace $\atrace_\pcinstr= (\rho_\pcinstr, \po_\pcinstr, \rfo_\pcinstr, \sto_\pcinstr, \cfo_\pcinstr)$ where every transaction $\atr \in \atrace$ is split into two  
transactions $\atrrd{\atr}\in \atrace_\pcinstr$ and $\atrwr{\atr} \in \atrace_\pcinstr$, and the dependency relations are straightforward adaptations, i.e., 
\begin{itemize}[topsep=3pt]
	\item $\po_\pcinstr$ is the smallest transitive relation that includes $(\atrrd{\atr},\atrwr{\atr})$ for every $\atr$, and $(\atrwr{\atr},\atrrd{\atr'})$ if $(\atr,\atr')\in \po$,
	\item 
	$(\atrwr{\atr'},\atrrd{\atr}) \in \rfo_\pcinstr$, $(\atrwr{\atr'},\atrwr{\atr}) \in \sto_\pcinstr$, and $(\atrrd{\atr'},\atrwr{\atr}) \in \cfo_\pcinstr$ if 
	$(\atr',\atr) \in \rfo$, $(\atr',\atr) \in \sto$, and $(\atr',\atr) \in \cfo$, respectively. 
\end{itemize}

\begin{wrapfigure}{r}{0.45\textwidth}
    \vspace{-20pt}
    \lstset{basicstyle=\ttfamily\scriptsize}
    \centering
    \scalebox{0.70}
    {
    \begin{tikzpicture}[->,>=stealth',shorten >=1pt,auto,node distance=4cm,
      semithick, transform shape,every text node part/.style={align=left}]
     \node[shape=rectangle ,draw=none,font=\large, label={left:$\atrrd{\atr_1}$}] at (-3,0)  (m)  {$[r1 = x]\ //0$};
     \node[shape=rectangle ,draw=none,font=\large, label={left:$\atrwr{\atr_1}$}] at (-3,-1.5) (m1)  {$[x = r1 + 1]$};
     \node[shape=rectangle ,draw=none,font=\large, label={right:$\atrrd{\atr_2}$}] at (0,0) (p){$[r2 = x]\ //0$};
     \node[shape=rectangle ,draw=none,font=\large, label={right:$\atrwr{\atr_2}$}] at (0,-1.5) (p1){$[x = r2 + 1]$};
     \begin{scope}[ every edge/.style={draw=black,very thick}]
     \path[->] (m1) edge[below] node {$\sto$} (p1);
     \path[->] (m) edge[left] node {$\po$} (m1);
     \path[->] (m) edge[below] node {$\cfo$} (p1);
     \path[->] (p) edge[above] node {$\cfo$} (m1);
     \path[->] (p) edge[left] node {$\po$} (p1);
    \end{scope}
    \end{tikzpicture}}
    \vspace{-7pt}
    \caption{A trace of the transformed LU program ($\mathsf{LU}_{\pcinstr}$).}
    \label{fig:litmus2Instr}
    \vspace{-20pt}
\end{wrapfigure}
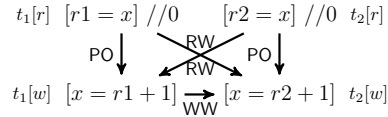
For instance, Figure \ref{fig:litmus2Instr} pictures the trace $\atrace_\pcinstr$ for the $\mathsf{LU}$ trace $\atrace$ given in Figure \ref{fig:litmus2}. 
For traces $\atrace$ of programs that contain singleton transactions, e.g., $\mathsf{SB}$ in Figure \ref{fig:litmus1}, $\atrace_\pcinstr$ coincides with $\atrace$.

Conversely, for a given trace $\atrace_\pcinstr= (\rho_\pcinstr, \po_\pcinstr, \rfo_\pcinstr, \sto_\pcinstr, \cfo_\pcinstr) \in \tracesconf_{\textsf{X}}(\aprog_\pcinstr)$ of a program $\aprog_\pcinstr$ under a semantics $\textsf{X}$, we define the trace $\atrace= (\rho, \po, \rfo, \sto, \cfo)$ where every two components $\atrrd{\atr}$ and $\atrwr{\atr}$ are merged into a transaction $\atr \in \atrace$. The dependency relations are defined in a straightforward way, e.g., if $(\atrwr{\atr'},\atrwr{\atr}) \in \sto_\pcinstr$ then  $(\atr',\atr) \in \sto$.

The following lemma shows that for any semantics $\textsf{X} \in \{\ccc,\ \pcc{},\ \sic{}\}$, if $\atrace \in \tracesconf_{\textsf{X}}(\aprog)$ for a program $\aprog$, then $\atrace_\pcinstr$ is a valid trace of $\aprog_\pcinstr$ under $\textsf{X}$, i.e., $\atrace_\pcinstr \in \tracesconf_{\textsf{X}}(\aprog_\pcinstr)$. Intuitively, this lemma shows that splitting transactions in a trace and defining dependency relations appropriately cannot introduce cycles in these relations and preserves the validity of the different consistency axioms.  

The proof of this lemma relies on constructing a causal order $\viso_\pcinstr$ and an arbitration order $\arbo_\pcinstr$ for the trace $\atrace_\pcinstr$ starting from the analogous relations in $\atrace$. In the case of $\ccc$, these are the smallest transitive relations such that:
\begin{itemize}[topsep=3pt]
    \item $\po_\pcinstr\subseteq \viso_\pcinstr\subseteq \arbo_\pcinstr$, and
    \item if $(\atr_{1},\atr_{2}) \in \viso$ then $(\atrwr{\atr_{1}},\atrrd{\atr_{2}}) \in \viso_\pcinstr$, and if $(\atr_{1},\atr_{2}) \in \arbo$ then $(\atrwr{\atr_{1}},\atrrd{\atr_{2}}) \in \arbo_\pcinstr$.
\end{itemize}
For \pcc{} and \sic{}, $\viso_\pcinstr$ must additionally satisfy: if $(\atr_{1},\atr_{2}) \in \arbo$, then $(\atrwr{\atr_{1}},\atrwr{\atr_{2}}) \in \viso_\pcinstr$. This is required in order to satisfy the axiom $\axprefix$, i.e., $\arbo_\pcinstr;\viso_\pcinstr \subset \viso_\pcinstr$, when $(\atrwr{\atr_{1}},\atrrd{\atr_{2}}) \in \arbo_\pcinstr$ and $(\atrrd{\atr_{2}},\atrwr{\atr_{2}}) \in \viso_\pcinstr$.

This construction ensures that $\viso_\pcinstr$ is a partial order and $\arbo_\pcinstr$ is a total order because $\viso$ is a partial order and $\arbo$ is a total order.
Also, based on the above rules, we have that: if $(\atrwr{\atr_{1}},\atrrd{\atr_{2}}) \in \viso_\pcinstr$ then $(\atr_{1},\atr_{2}) \in \viso$, and similarly, if $(\atrwr{\atr_{1}},\atrrd{\atr_{2}}) \in \arbo_\pcinstr$ then $(\atr_{1},\atr_{2}) \in \arbo$.

\begin{lemma} \label{lem:Transform}
If $\atrace \in \tracesconf_{\textsf{X}}(\aprog)$, then 
$\atrace_\pcinstr \in \tracesconf_{\textsf{X}}(\aprog_\pcinstr)$.
\end{lemma}

Before presenting a strengthening of Lemma \ref{lem:Transform} when $\textsf{X}$ is \ccc{}, we give an important characterization of \ccc{} traces. This characterization is stated in terms of  acyclicity properties.


\begin{lemma} \label{lem:cycles} 
$\atrace$ is a trace under \ccc{} iff $\arbo_{0}^{+}$ and $\viso_{0}^{+};\cfo$ are acyclic ($\arbo_{0}$ and $\viso_{0}$ are defined in Table \ref{fig:defaxioms}). 
\end{lemma}

Next we show that a trace $\atrace$ of a program $\aprog$ is \ccc{} iff the corresponding trace $\atrace_\pcinstr$ of $\aprog_\pcinstr$ is \ccc{} as well. This result is based on the observation that cycles in $\arbo_{0}^{+}$ or $\viso_{0}^{+};\cfo$ cannot be broken by splitting transactions. 

\begin{lemma} \label{lem:CcCc}
    A trace $\atrace$ of $\aprog$ is \ccc{} iff the corresponding trace $\atrace_\pcinstr$ of $\aprog_\pcinstr$ is \ccc{}.
\end{lemma} 
The following lemma shows that a trace $\atrace$ is \pcc{} iff the corresponding trace $\atrace_\pcinstr$ is \serc{}. 
The if direction in the proof is based on constructing a causal order $\viso$ and an arbitration order $\arbo$ for the trace $\atrace$ from the arbitration order $\arbo_\pcinstr$ in $\atrace_\pcinstr$ (since $\atrace_\pcinstr$ is a trace under serializability $\viso_\pcinstr$ and $\arbo_\pcinstr$ coincide). These are the smallest transitive relations such that:
\begin{itemize}[topsep=3pt]
    \item if $(\atrwr{\atr_{1}},\atrrd{\atr_{2}}) \in \arbo_\pcinstr$ then $(\atr_{1},\atr_{2}) \in \viso$,
    \item if $(\atrwr{\atr_{1}},\atrwr{\atr_{2}}) \in \arbo_\pcinstr$ then $(\atr_{1},\atr_{2}) \in \arbo$\footnote{If $\atrwr{\atr_{1}}$ is empty ($\atr_1$ is read-only), then we set $(\atr_{1},\atr_{2}) \in \arbo$ if $(\atrrd{\atr_{1}},\atrwr{\atr_{2}}) \in \viso_\pcinstr$. If $\atrwr{\atr_{2}}$ is empty, then $(\atr_{1},\atr_{2}) \in \arbo$ if $(\atrwr{\atr_{1}},\atrrd{\atr_{2}}) \in \viso_\pcinstr$. If both $\atrwr{\atr_{1}}$ and $\atrwr{\atr_{2}}$ are empty, then $(\atr_{1},\atr_{2}) \in \arbo$ if $(\atrrd{\atr_{1}},\atrrd{\atr_{2}}) \in \viso_\pcinstr$.}.
\end{itemize}
The only-if direction is based on the fact that any cycle in the dependency relations of $\atrace$ that is admitted under \pcc{} (characterized in Lemma \ref{lem:pccycles}) is ``broken'' by splitting transactions. Also, splitting transactions cannot introduce new cycles that do not originate in $\atrace$.
\vspace{-3pt}
\begin{lemma}\label{lem:PcSer}
    A trace $\atrace$ is \pcc{} iff $\atrace_\pcinstr$ is \serc{}
\end{lemma} 
\vspace{-3pt}
The lemmas above are used to prove Theorem \ref{them:RobCcPc} as follows:

\medskip
\noindent
\textsc{Proof} of Theorem \ref{them:RobCcPc}:
For the if direction, assume by contradiction that $\aprog$ is not robust against \ccc{} relative to \pcc{}. 
Then, there must exist a trace $\atrace \in \tracesconf_{\ccc{}}(\aprog) \setminus \tracesconf_{\pcc{}}(\aprog)$. Lemmas~\ref{lem:CcCc} and~\ref{lem:PcSer} imply that the corresponding trace $\atrace_\pcinstr$ of $\aprog_\pcinstr$ is \ccc{} and not \serc{}. Thus, $\aprog_\pcinstr$ is not robust against \ccc{} relative to \serc{}. The only-if direction is proved similarly.
\hfill $\Box$

\medskip
Robustness against \ccc{} relative to \serc{} has been shown to be reducible in polynomial time to the reachability problem under \serc{}~\cite{DBLP:conf/concur/BeillahiBE19}. Given a program $\aprog$ and a control location $\ell$, the reachability problem under \serc{} asks whether there exists an execution of $\aprog$ under \serc{} that reaches $\ell$.
Therefore, as a corollary of Theorem \ref{them:RobCcPc}, we obtain the following:


\vspace{-5pt}
\begin{corollary}\label{thm:RobViolCcPc}
 Checking robustness against \ccc{} relative to \pcc{} is reducible to the reachability problem under \serc{} in polynomial time.
\end{corollary}
\vspace{-3pt}
In the following we discuss the complexity of this problem in the case of finite-state programs (bounded data domain). The upper bound follows from Corollary~\ref{thm:RobViolCcPc} and 
standard results about the complexity of the reachability problem under sequential consistency, which extend to \serc{}, with a bounded~\cite{DBLP:conf/focs/Kozen77} or parametric number of processes~\cite{DBLP:journals/tcs/Rackoff78}. For the lower bound, given an instance $(\aprog,\ell)$ of the reachability problem under sequential consistency, we construct a program $\aprog'$ where each statement $s$ of $\aprog$ is executed in a different transaction that guards\footnote{That is, the transaction is of the form [lock; $s$; unlock]} the execution of $s$ using a global lock (the lock can be implemented in our programming language as usual, e.g., using a busy wait loop for locking), and where reaching the location $\ell$ enables the execution of a ``gadget'' that corresponds to the $\mathsf{SB}$ program in Figure~\ref{fig:litmus1}. Executing each statement under a global lock ensures that every execution of $\aprog'$ under $\ccc$ is serializable, and faithfully represents an execution of $\aprog$ under sequential consistency. Moreover, $\aprog$ reaches $\ell$ iff $\aprog'$ contains a robustness violation, which is due to the $\mathsf{SB}$ execution.
 
\vspace{-5pt}
\begin{corollary}\label{corol:PCRobcomplexity}
Checking robustness of a program with a fixed number of variables and bounded data domain against \ccc{} relative to \pcc{} is PSPACE-complete when the number of processes is bounded and EXPSPACE-complete, otherwise.
\end{corollary}
\vspace{-5pt}
\vspace{-15pt}
\section{Robustness Against \pcc{} Relative to \sic{}}
\label{sec:robustness}
\vspace{-5pt}
In this section, we show that checking robustness against \pcc{} relative to \sic{} can be reduced in polynomial time to a reachability problem under the \serc{} semantics. We reuse the program transformation from the previous section that allows to simulate \pcc{} behaviors on top of \serc{}, and additionally, we provide a characterization of traces that distinguish the \pcc{} semantics from \sic. We use this characterization to define an instrumentation (monitor) that is able to detect if a program under \pcc{} admits such traces.

We show that the happens-before cycles in a robustness violation (against \pcc{} relative to \sic{}) must contain a $\sto$ dependency followed by a $\cfo$ dependency, and they should not contain two successive $\cfo$ dependencies. This follows from the fact that every happens-before cycle in a \pcc{} trace must contain either two successive $\cfo$ dependencies, or a $\sto$ dependency followed by a $\cfo$ dependency. Otherwise, the happens-before cycle would imply a cycle in the arbitration order. Then, any trace under \pcc{} where all its simple happens-before cycles contain two successive $\cfo$ dependencies is possible under \sic{}. 
%
%
For instance, the trace of the non-robust $\mathsf{LU}$ execution in Figure~\ref{fig:litmus2} contains $\sto$ dependency followed by a $\cfo$ dependency and does not contain two successive $\cfo$ dependencies which is disallowed \sic{}, while the trace of the robust $\mathsf{WS}$ execution in Figure~\ref{fig:litmus3} contains two successive $\cfo$ dependencies.  
As a first step, we prove the following theorem characterizing traces that are allowed under both \pcc{} and \sic{}. 

\vspace{-2pt}
\begin{theorem} \label{them:RobPcSiLem}
A program $\aprog$ is robust against \pcc{} relative to \sic{} iff every happens-before cycle in a trace of $\aprog$ under \pcc{} contains two successive $\cfo$ dependencies.
\end{theorem}
\vspace{-2pt}
Before giving the proof of the above theorem, we state several intermediate results that characterize cycles in \pcc{} or \sic{} traces. First, we show that every \pcc{} trace in which all simple happens-before cycles contain two successive $\cfo$ is also a \sic{} trace.

\vspace{-3pt}
\begin{lemma}\label{lem:pcsicycles}
If a trace $\atrace$ is \pcc{} and all happens-before cycles in $\atrace$ contain two successive $\cfo$ dependencies, then $\atrace$ is \sic{}.
\end{lemma}
\vspace{-3pt}

The proof of Theorem \ref{them:RobPcSiLem} also relies on the following lemma that characterizes happens-before cycles permissible under \sic{}.
\vspace{-3pt}
\begin{lemma}\label{lem:sicycles}\cite{DBLP:journals/tods/CahillRF09,DBLP:conf/concur/0002G16}
    If a trace $\atrace$ is \sic{}, then all its happens-before cycles must contain two successive $\cfo$ dependencies.
\end{lemma}
\vspace{-3pt}
\noindent
\textsc{Proof} of Theorem \ref{them:RobPcSiLem}:
    For the only-if direction, if $\aprog$ is robust against \pcc{} relative to \sic{} then every trace $\atrace$ of $\aprog$ under \pcc{} is \sic{} as well. Therefore, by Lemma \ref{lem:sicycles}, all cycles in $\atrace$ contain two successive $\cfo$ which concludes the proof of this direction. 
    For the reverse, let $\atrace$ be a trace of $\aprog$ under \pcc{} such that all its happens-before cycles 
    contain two successive $\cfo$. Then, by Lemma \ref{lem:pcsicycles}, we have that $\atrace$ is \sic{}.  
    Thus, every trace $\atrace$ of $\aprog$ under \pcc{} is \sic{}.    
\hfill $\Box$

\medskip
Next, we present an important lemma that characterizes happens before cycles possible under the \pcc{} semantics. 
This is a strengthening of a result in~\cite{DBLP:conf/concur/0002G16} which shows that all happens before cycles under \pcc{} must have two successive dependencies in $\{\cfo,\sto\}$ and at least one $\cfo$. We show that the two successive dependencies cannot be $\cfo$ followed $\sto$, or two successive $\sto$.
\vspace{-3pt}
\begin{lemma} \label{lem:pccycles}
    If a trace $\atrace$ is \pcc{} then all happens-before cycles in $\atrace$ must contain either two successive $\cfo$ dependencies or a $\sto$ dependency followed by a $\cfo$ dependency.
\end{lemma}
\vspace{-5pt}
Combining the results of Theorem \ref{them:RobPcSiLem} and Lemmas~\ref{lem:PcSer} and~\ref{lem:pccycles}, we obtain the following characterization of traces which violate robustness against \pcc{} relative to \sic{}.

\begin{theorem} \label{corol:pcsi}
    A program $\aprog$ is not robust against \pcc{} relative to \sic{} iff there exists a trace $\atrace_\pcinstr$ of $\aprog_\pcinstr$ under \serc{} such that the trace $\atrace$ obtained by merging\footnote{This transformation has been defined at the beginning of Section~\ref{sec:CCPCrobustness}.} read and write transactions in $\atrace_\pcinstr$ contains a happens-before cycle that does not contain two successive $\cfo$ dependencies, and it contains a $\sto$ dependency followed by a $\cfo$ dependency.      
\end{theorem} 
\vspace{-3pt}  
The results above enable a reduction from checking robustness against \pcc{} relative to \sic{} to a reachability problem under the \serc{} semantics. For a program $\aprog$, we define an instrumentation denoted by $\sem{\aprog}$, such that $\aprog$ is not robust against \pcc{} relative to \sic{} iff $\sem{\aprog}$ violates an assertion under \serc{}. The instrumentation consists in rewriting every transaction of $\aprog$ as shown in Figure~\ref{Figure:Instr}. 

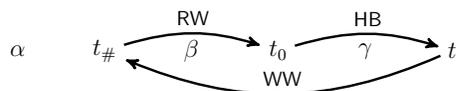
\begin{wrapfigure}{r}{0.5\textwidth}
  \vspace{-33pt}
\begin{center}
\scalebox{0.77}
{
\begin{tikzpicture}[->,>=stealth',shorten >=1pt,auto,node distance=4cm,
    semithick, transform shape,every text node part/.style={align=left}]
   \node[shape=rectangle ,draw=none,font=\large] at (0,0)  (m)  {$\alpha$};
   \node[shape=rectangle ,draw=none,font=\large] at (1.5,0)  (n) {$\atr_{\instr}$};
   \node[shape=rectangle ,draw=none,font=\large] at (3,0)  (n1) {$\beta$};
   \node[shape=rectangle ,draw=none,font=\large] at (4.5,0)  (n2) {$\atr_{0}$};
   \node[shape=rectangle ,draw=none,font=\large] at (6,0)  (n3) {$\gamma$};
   \node[shape=rectangle ,draw=none,font=\large] at (7.5,0)  (n4) {$\atr$};
   \begin{scope}[ every edge/.style={draw=black,very thick}]
   \path[->] (n) edge[bend left=17] node {$\cfo$} (n2);
   \path[->] (n2) edge[bend left=17] node {$\hbo$} (n4);
   \path[->] (n4) edge[bend left=23, above] node {$\sto$} (n);
    \end{scope}
\end{tikzpicture}}   
\end{center}
\vspace{-10pt}
\caption{Execution simulating a violation to robustness against \pcc{} relative to \sic{}.}
\label{fig:violationInstr}
\vspace{-15pt}
\end{wrapfigure}
The instrumentation $\sem{\aprog}$ running under $\serc$ simulates the \pcc{} semantics of $\aprog$ using the same idea of decoupling the execution of the read part of a transaction from the write part. It violates an assertion when it simulates a \pcc{} trace containing a happens-before cycle as in Theorem \ref{corol:pcsi}. The execution corresponding to this trace has the shape given in Figure~\ref{fig:violationInstr}, where $\atr_{\instr}$ is the transaction that occurs between the $\sto$ and the $\cfo$ dependencies, and every transaction executed after $\atr_{\instr}$ (this can be a full transaction in $\aprog$, or only the read or write part of a transaction in $\aprog$) is related by a happens-before path to $\atr_{\instr}$ (otherwise, the execution of this transaction can be reordered to occur before $\atr_{\instr}$). A transaction in $\aprog$ can have its read part included in $\alpha$ and the write part included in $\beta$ or $\gamma$. Also, $\beta$ and $\gamma$ may contain transactions in $\aprog$ that executed only their read part. It is possible that $\atr_{0} = \atr$, $\beta=\gamma=\epsilon$, and $\alpha = \epsilon$ (the $\mathsf{LU}$ program shown in Figure \ref{fig:litmus2} is an example where this can happen). The instrumentation uses auxiliary variables to track happens-before dependencies, which are explained below.

%

\lstset{basicstyle=\ttfamily\scriptsize,numbers=left,
      stepnumber=1,numberblanklines=false,mathescape=true,escapechar=$,morekeywords={assert,assume}}

\begin{figure}[!ht]
\footnotesize
\begin{minipage}{\linewidth}
Transaction ``\plog{begin} $\langle$read$\rangle^{*}$ $\langle$test$\rangle^{*}$ $\langle$write$\rangle^{*}$ \plog{commit}'' is rewritten to:
\end{minipage}
\begin{minipage}{0.655\linewidth}
\begin{lstlisting}
if ( !done$_\#$ )
  if ( * )
    begin <read>$^{*}$ <test>$^{*}$ commit $\label{ln:part11}$
    if ( !done$_\#$ )
      begin <write>$^{*}$ commit $\label{ln:part12}$
    else 
      $\mathcal{I}$(begin) ($\mathcal{I}$(<write>))$^{*}$ $\mathcal{I}$(commit)$\label{ln:part31}$
  else
     begin ($\mathcal{I}_\#$(<read>))$^{*}$ <test>$^{*}$ ($\mathcal{I}_\#$(<write>))$^{*}$ $\mathcal{I}_\#$(commit)$\label{ln:part21}$
     assume false;
else if ( * )
  rdSet' := $\emptyset$;
  wrSet' := $\emptyset$;
  $\mathcal{I}$(begin) ($\mathcal{I}$(<read>))$^{*}$ <test>$^{*}$ $\mathcal{I}$(commit)$\label{ln:part32}$
  $\mathcal{I}$(begin) ($\mathcal{I}$(<write>))$^{*}$ $\mathcal{I}$(commit)$\label{ln:part33}$
\end{lstlisting}
\end{minipage}\hfill
\begin{minipage}{0.305\linewidth}
$\mathcal{I}_\#$( r := x ):
\begin{lstlisting}[xleftmargin=2mm,firstnumber=16]
r := x; $\label{ln:delay1}$
hbR['x'] := 0;
rdSet := rdSet $\cup$ { 'x' };
\end{lstlisting}
$\mathcal{I}_\#$( x := e ):
\begin{lstlisting}[xleftmargin=2mm,firstnumber=19]
if ( varW == $\bot$ and * )
  varW := 'x';
\end{lstlisting}
$\mathcal{I}_\#$( commit ):
\begin{lstlisting}[xleftmargin=2mm,firstnumber=21]
assume ( varW != $\bot$ )
done$_\#$ := true $\label{ln:delay2}$
\end{lstlisting}
\end{minipage}

\vspace{1mm}
\begin{minipage}{0.5\linewidth}
$\mathcal{I}$( begin ):
\begin{lstlisting}[xleftmargin=3mm,firstnumber=23]
begin
hb := $\bot$
if ( hbP != $\bot$ and hbP < 2 )
  hb := 0;
else if ( hbP = 2 )
  hb := 2;
\end{lstlisting}
$\mathcal{I}$(  commit ):
\begin{lstlisting}[xleftmargin=3mm,firstnumber=29]
assume ( hb != $\bot$ ) $\label{ln:assume}$
assert ( hb == 2 or varW $\not\in$ wrSet' ); $\label{ln:assert}$
if ( hbP == $\bot$ or hbP > hb ) $\label{ln:hbupdates1}$
  hbP = hb;
for each 'x' $\in$ wrSet'
  if ( hbW['x'] == $\bot$ or hbW['x'] > hb )
    hbW['x'] = hb; 
for each 'x' $\in$ rdSet'
  if ( hbR['x'] == $\bot$ or hbR['x'] > hb )
    hbR['x'] = hb;  $\label{ln:hbupdates2}$
rdSet := rdSet $\cup$ rdSet'; $\label{ln:rdSetUpdate}$
wrSet := wrSet $\cup$ wrSet'; $\label{ln:wrSetUpdate}$
commit
\end{lstlisting}
\end{minipage}
\hfill
\begin{minipage}{0.45\linewidth}
$\mathcal{I}$(  r := x ):
\begin{lstlisting}[xleftmargin=3mm,firstnumber=42]
r := x;
rdSet' := rdSet' $\cup$ { 'x' };
if ( 'x' $\in$ wrSet ) $\label{ln:rdHBCont}$
  if ( hbW['x'] != 2 ) 
    hb := 0 $\label{ln:rdHBConthb1}$
  else if ( hb == $\bot$ )
    hb := hbW['x'] $\label{ln:rdHBConthb2}$
\end{lstlisting}
$\mathcal{I}$(  x := e ):
\begin{lstlisting}[xleftmargin=3mm,firstnumber=49]
x := e;
wrSet' := wrSet' $\cup$ { 'x' };
if ( 'x' $\in$ wrSet ) $\label{ln:wrHBCont1}$
  if ( hbW['x'] != 2 )
    hb := 0
  else if ( hb == $\bot$ )
    hb := hbW['x']
if ( 'x' $\in$ rdSet ) $\label{ln:wrHBCont2}$
  if ( hb = $\bot$ or hb > hbR['x'] + 1 )
    hb := min(hbR['x'] + 1,2)
\end{lstlisting}
\end{minipage}
\vspace{-10pt}
   \normalsize
   \caption{A program instrumentation for checking robustness against \pcc{} relative to \sic{}. The auxiliary variables used by the instrumentation are shared variables, except for \texttt{hbP}, \texttt{rdSet'}, and \texttt{wrSet'}, which are process-local variables, and they are initially set to $\bot$. This instrumentation uses program constructs which can be defined as syntactic sugar from the syntax presented in Section~\ref{sec:consistency}, e.g., if-then-else statements (outside transactions).}
   \label{Figure:Instr}
   \vspace{-10pt}
   \end{figure}

The instrumentation executes (incomplete) transactions without affecting the auxiliary variables (without tracking happens-before dependencies) (lines~\ref{ln:part11} and \ref{ln:part12}) until a non-deterministically chosen point in time when it declares the current transaction as the candidate for $\atr_{\instr}$ (line~\ref{ln:part21}). Only one candidate for $\atr_{\instr}$ can be chosen during the execution. This transaction executes only its reads and it chooses non-deterministically a variable that it could write as a witness for the $\sto$ dependency (see lines~\ref{ln:delay1}-\ref{ln:delay2}). The name of this variable is stored in a global variable \texttt{varW} (see the definition of $\mathcal{I}_\#$( x := e )).
The writes are \emph{not} applied on the shared memory. Intuitively, $\atr_{\instr}$ should be thought as a transaction whose writes are delayed for later, after transaction $\atr$ in Figure~\ref{fig:violationInstr} executed. The instrumentation checks that $\atr_{\instr}$ and $\atr$ can be connected by some happens-before path that includes the $\cfo$ and $\sto$ dependencies, and that does not contain two consecutive $\cfo$ dependencies. If it is the case, it violates an assertion at the commit point of $\atr$. Since the write part of $\atr_{\instr}$ is intuitively delayed to execute after $\atr$, the process executing $\atr_{\instr}$ is disabled all along the execution (see the \texttt{assume false}).

After choosing the candidate for $\atr_{\instr}$, the instrumentation uses the auxiliary variables for tracking happens-before dependencies. Therefore, \texttt{rdSet} and \texttt{wrSet} record variables read and written, respectively, by transactions that are connected by a happens-before path to $\atr_{\instr}$ (in a trace of $\aprog$). This is ensured by the assume at line~\ref{ln:assume}. During the execution, the variables read or written by a transaction\footnote{These are stored in the local variables \texttt{rdSet'} and \texttt{wrSet'} while the transaction is running.} that writes a variable in \texttt{rdSet} (see line~\ref{ln:wrHBCont2}), or reads or writes a variable in \texttt{wrSet} (see lines~\ref{ln:rdHBCont} and~\ref{ln:wrHBCont1}), will be added to these sets (see lines~\ref{ln:rdSetUpdate} and~\ref{ln:wrSetUpdate}). 
Since the variables that $\atr_{\instr}$ writes in $\aprog$ are not recorded in \texttt{wrSet}, these happens-before paths must necessarily start with a $\cfo$ dependency (from $\atr_{\instr}$). When the assertion fails (line~\ref{ln:assert}), the condition \texttt{varW} $\in$ \texttt{wrSet}' ensures that the current transaction has a $\sto$ dependency towards the write part of $\atr_{\instr}$ (the current transaction plays the role of $\atr$ in Figure~\ref{fig:violationInstr}).


The rest of the instrumentation checks that there exists a happens-before path from $\atr_{\instr}$ to $\atr$ that does not include two consecutive $\cfo$ dependencies, called a \sic{}$_{\neg}$ path. This check is based on the auxiliary variables whose name is prefixed by \texttt{hb} and which take values in the domain $\{\bot,0,1,2\}$ ($\bot$ represents the initial value). 
Therefore, 
\begin{itemize}[topsep=3pt]
	\item \texttt{hbR['x']} (resp., \texttt{hbW['x']}) is 0 iff there exists a transaction $\atr'$ that reads \texttt{x} (resp., writes to \texttt{x}), such that there exists a \sic{}$_{\neg}$ path from $\atr_{\instr}$ to $\atr'$ that ends with a dependency which is \emph{not} $\cfo$,
	\item \texttt{hbR['x']} (resp., \texttt{hbW['x']}) is 1 iff there exists a transaction $\atr'$ that reads \texttt{x} (resp., writes to \texttt{x}) that is connected to $\atr_{\instr}$ by a \sic{}$_{\neg}$  path, and \emph{every} \sic{}$_{\neg}$ path from $\atr_{\instr}$ to a transaction $\atr''$ that reads \texttt{x} (resp., writes to \texttt{x}) ends with an $\cfo$ dependency,
	\item \texttt{hbR['x']} (resp., \texttt{hbW['x']}) is 2 iff there exists no \sic{}$_{\neg}$ path from $\atr_{\instr}$ to a transaction $\atr'$ that reads \texttt{x} (resp., writes to \texttt{x}).
\end{itemize}
The local variable \texttt{hbP} has the same interpretation, except that $\atr'$ and $\atr''$ are instantiated over transactions in the same process (that already executed) instead of transactions that read or write a certain variable. Similarly, the variable \texttt{hb} is a particular case where $\atr'$ and $\atr''$ are instantiated to the current transaction. The violation of the assertion at line~\ref{ln:assert} implies that \texttt{hb} is 0 or 1, which means that there exists a \sic{}$_{\neg}$ path from $\atr_{\instr}$ to $\atr$.

During each transaction that executes after $\atr_{\instr}$, the variable \texttt{hb} characterizing happens-before paths that end in this transaction is updated every time a new happens-before dependency is witnessed (using the values of the other variables). For instance, when witnessing a $\rfo$ dependency (line~\ref{ln:rdHBCont}), if there exists a \sic{}$_{\neg}$ path to a transaction that writes to \texttt{x}, then the path that continues with the $\rfo$ dependency towards the current transaction is also a \sic{}$_{\neg}$ path, and the last dependency of this path is not $\cfo$. Therefore, \texttt{hb} is set to 0 (see line~\ref{ln:rdHBConthb1}). Otherwise, if every path to a transaction that writes to \texttt{x} is not a \sic{}$_{\neg}$ path, then every path that continues to the current transaction (by taking the $\rfo$ dependency) remains a non \sic{}$_{\neg}$ path, and \texttt{hb} is set to the value of \texttt{hbW[`x`]}, which is 2 in this case (see line~\ref{ln:rdHBConthb2}). Before ending a transaction, the value of \texttt{hb} can be used to modify the \texttt{hbR}, \texttt{hbW}, and \texttt{hbP} variables, but only if those variables contain bigger values (see lines~\ref{ln:hbupdates1}--\ref{ln:hbupdates2}).

The correctness of the instrumentation is stated in the following theorem.

\begin{theorem}\label{them:RobPcSiInstr}
A program $\aprog$ is robust against \pcc{} relative to \sic{} iff the instrumentation in Figure~\ref{Figure:Instr} does not violate an assertion when executed under \serc{}.
\end{theorem}

Theorem~\ref{them:RobPcSiInstr} implies the following complexity result for finite-state programs. The lower bound is proved similarly to the case \ccc{} vs \pcc{}.


\begin{corollary}\label{corol:SIRobcomplexity}
Checking robustness of a program with a fixed number of variables and bounded data domain against \pcc{} relative to \sic{} is PSPACE-complete when the number of processes is bounded and EXPSPACE-complete, otherwise.
\end{corollary}

Checking robustness against \ccc{} relative to \sic{} can be also shown to be reducible (in polynomial time) to a reachability problem under \serc{} by combining the results of checking robustness against \ccc{} relative to \pcc{} and \pcc{} relative to \sic{}. 

\begin{theorem} \label{them:RobCcSi}
A program $\aprog$ is robust against \ccc{} relative to \sic{} iff $\aprog$ is robust against \ccc{} relative to \pcc{} and $\aprog$ is robust against \pcc{} relative to \sic{}.
\end{theorem}

\begin{remark} \label{rem:robustness}
  Our reductions of robustness checking to reachability apply to an extension of our programming language where the number of processes is unbounded and each process can execute an arbitrary number of times a statically known set of transactions. This holds because the instrumentation in Figure~\ref{Figure:Instr} and the one in~\cite{DBLP:conf/concur/BeillahiBE19} (for the case \ccc{} vs. \serc{}) consist in adding a set of instructions that manipulate a fixed set of process-local or shared variables, which do not store process or transaction identifiers. These reductions extend also to SQL queries that access unbounded size tables. Rows in a table can be interpreted as memory locations (identified by primary keys in unbounded domains, e.g., integers), and SQL queries can be interpreted as instructions that read/write a set of locations in one shot. These possibly unbounded sets of locations can be represented symbolically using the conditions in the SQL queries (e.g., the condition in the WHERE part of a SELECT). The instrumentation in Figure 6 needs to be adapted so that read and write sets are updated by adding sets of locations for a given instruction (represented symbolically as mentioned above). 
\end{remark}
    
\vspace{-15pt}
\section{Proving Robustness Using Commutativity Dependency Graphs}
\label{sec:commutativitygraph}

We describe an approximated technique for proving robustness, which leverages the concept of left/right mover in Lipton's reduction theory~\cite{DBLP:journals/cacm/Lipton75}. This technique reasons on the \emph{commutativity dependency graph}~\cite{DBLP:conf/cav/BeillahiBE19} associated to the transformation $\aprog_\pcinstr$ of an input program $\aprog$ that allows to simulate the \pcc{} semantics under serializability (we use a slight variation of the original definition of this class of graphs).
We characterize robustness against \ccc{} relative to \pcc{} and \pcc{} relative to \sic{} in terms of certain properties that (simple) cycles in this graph must satisfy. 

We recall the concept of movers and the definition of commutativity dependency graphs. 
Given a program $\aprog$ and a trace $\atrace = \atr_1\cdot \ldots\cdot \atr_n \in \tracesconf_{\serc{}}(\aprog)$ of  $\aprog$ under serializability, we say that $\atr_i \in \atrace$ \emph{moves right (resp., left)} in $\atrace$ if $\atr_1\cdot \ldots\cdot \atr_{i-1}\cdot \atr_{i+1}\cdot \atr_i\cdot \atr_{i+2}\cdot \ldots\cdot \atr_n$ (resp., $\atr_1\cdot\ldots\cdot\atr_{i-2}\cdot\atr_i\cdot\atr_{i-1}\cdot\atr_{i+1}\cdot\ldots\cdot\atr_n$) is also a valid execution of $\aprog$, $\atr_i$ and $\atr_{i+1}$ (resp., $\atr_{i-1}$) are executed by distinct processes, and both traces reach the same end state. A transaction $\atr \in \trsaprog{\aprog}$ is not a right (resp., left) mover iff there exists a trace $\atrace \in \tracesconf_{\serc{}}(\aprog)$ such that $\atr \in \atrace$ and $\atr$ doesn't move right (resp., left) in $\atrace$. Thus, when a transaction $\atr$ is \emph{not} a right mover then there must exist another transaction $\atr' \in \atrace$ which caused $\atr$ to not be permutable to the right (while preserving the end state). Since $\atr$ and $\atr'$ do not commute, then this must be because of either a write-read, write-write, or a read-write dependency relation between the two transactions. We say that $\atr$ is not a right mover because of $\atr'$ and a dependency relation that is either write-read, write-write, or read-write. Notice that when $\atr$ is not a right mover because of $\atr'$ then $\atr'$ is not a left mover because of $\atr$. 

We define $\mrfo$ as a binary relation between transactions such that $(\atr,\atr') \in \mrfo$ when $\atr$ is \emph{not} a right mover because of $\atr'$ and a write-read dependency ($t'$ reads some value written by $t$). We define the relations $\msto$ and $\mcfo$ corresponding to write-write and read-write dependencies in a similar way. 
We call $\mrfo$, $\msto$, and $\mcfo$, \emph{non-mover} relations.

The \emph{commutativity dependency graph} of a program $\aprog$ is a graph where vertices represent transactions in $\aprog$. Two vertices are linked by a program order edge if the two transactions are executed by the same process. The other edges in this graph represent the ``non-mover'' relations $\mrfo$, $\msto$, and $\mcfo$. 
Two vertices that represent the two components $\atrwr{\atr}$ and $\atrrd{\atr}$ of the same transaction $\atr$ (already linked by $\po$ edge) are also linked by an undirected edge labeled by $\sametro$ (same-transaction relation). 

\begin{wrapfigure}{r}{0.53\textwidth}
    \vspace{-25pt}
    \lstset{basicstyle=\ttfamily\scriptsize}
    \centering
    \resizebox{!}{2.2cm}{
    \begin{tikzpicture}[->,>=stealth',shorten >=1pt,auto,node distance=4cm,
      semithick, transform shape,every text node part/.style={align=left}]
     \node[shape=rectangle ,draw=none,font=\large, label={left:$\atrwr{\atr 1}$}] at (-3.2,0)  (m)  {$[x = 1]$};
     \node[shape=rectangle ,draw=none,font=\large, label={left:$\atrwr{\atr 2}$}] at (-3.2,-2) (m1)  {$[y = 1]$};
     \node[shape=rectangle ,draw=none,font=\large, label={right:$\atrrd{\atr 3}$}] at (0,0) (p){$[r1 = y]$};
     \node[shape=rectangle ,draw=none,font=\large, label={right:$\atrrd{\atr 4}$}] at (0,-2) (p1){$[r2 = x]$};
     \begin{scope}[ every edge/.style={draw=black,very thick}]
     \path[->] (m) edge[left] node {$\po$} (m1);
     \path[->] (p1) edge[bend left, above] node[xshift=2.9mm,yshift=-1.3mm] {$\mcfo$} (m);
     \path[->] (m) edge[bend left, below] node {$\mrfo$} (p1);
     \path[->] (m1) edge[bend left, below] node[xshift=2.5mm,yshift=0.5mm] {$\mrfo$} (p);
     \path[->] (p) edge[bend left, above] node[xshift=-1.2mm,yshift=0.3mm] {$\mcfo$} (m1);
     \path[->] (p) edge[right] node {$\po$} (p1);
    \end{scope}
    \end{tikzpicture}}
    \vspace{-3pt}
    \caption{The commutativity dependency graph of the $\mathsf{MP}_{\pcinstr}$ program.}
    \label{fig:litmus4CDG}
    \vspace{-15pt}
\end{wrapfigure}
Our results about the robustness of a program $\aprog$ are stated over a slight variation of the commutativity dependency graph of $\aprog_\pcinstr$ (where a transaction is either read-only or write-only). This graph contains additional undirected edges that link every pair of transactions $\atrrd{\atr}$ and $\atrwr{\atr}$ of $\aprog_\pcinstr$ that were originally components of the same transaction $\atr$ in $\aprog$.
Given such a commutativity dependency graph, the robustness of $\aprog$ is implied by the absence of cycles of specific shapes. These cycles can be seen as an abstraction of potential robustness violations for the respective semantics (see Theorem~\ref{them:MovRobCcPc} and Theorem~\ref{them:MovRobPcSi}).
Figure \ref{fig:litmus4CDG} pictures the commutativity dependency graph for the $\mathsf{MP}$ program. Since every transaction in $\mathsf{MP}$ is singleton, the two programs $\mathsf{MP}$ and $\mathsf{MP}_{\pcinstr}$ coincide.  

Using the characterization of robustness violations against \ccc{}  relative to \serc{} from~\cite{DBLP:conf/concur/BeillahiBE19} and the reduction in Theorem~\ref{them:RobCcPc}, we obtain the following result concerning the robustness against \ccc{} relative to \pcc{}.

\begin{theorem} \label{them:MovRobCcPc}
    Given a program $\aprog$, if the commutativity dependency graph of the program $\aprog_\pcinstr$ does not contain a simple cycle formed by $\atr_1$ $\cdots$ $\atr_i$ $\cdots$ $\atr_n$ such that:
    \begin{itemize}[topsep=3pt]
        \item $(\atr_n,\atr_1) \in \mcfo$;
        \item $(\atr_j, \atr_{j+1}) \in (\po \cup \rfo)^{*}$, for $j \in [1,i-1]$;
        \item $(\atr_i,\atr_{i+1}) \in (\mcfo \cup \msto)$;
        \item $(\atr_j,\atr_{j+1}) \in (\mcfo \cup \msto \cup \mrfo \cup \po)$, for $j \in [i+1,n-1]$.
    \end{itemize}
    then $\aprog$ is robust against \ccc{} relative to \pcc{}.
\end{theorem}
Next we give the characterization of commutativity dependency graphs required for proving robustness against \pcc{} relative to \sic{}.

\begin{theorem} \label{them:MovRobPcSi}
    Given a program $\aprog$, if the commutativity dependency graph of the program $\aprog_\pcinstr$ does not contain a simple cycle formed by $\atr_1$ $\cdots$ $\atr_n$ such that:
    \begin{itemize}[topsep=3pt]
        \item $(\atr_n,\atr_1) \in \msto$, $(\atr_1,\atr_2) \in \sametro$, and  $(\atr_2,\atr_3) \in \mcfo$;
        \item $(\atr_j,\atr_{j+1}) \in (\mcfo \cup \msto \cup \mrfo \cup \po \cup \sametro)^{*}$, for $j \in [3,n-1]$;
        \item $\forall\ j \in [2,n-2].$
        \begin{itemize}
            \item $\mbox{if }(\atr_j,\atr_{j+1}) \in \mcfo\mbox{ then }(\atr_{j+1},\atr_{j+2}) \in (\mrfo \cup \po \cup \msto)$;
            \item $\mbox{if }(\atr_{j+1},\atr_{j+2}) \in \mcfo\mbox{ then }(\atr_{j},\atr_{j+1}) \in (\mrfo \cup \po)$.
        \end{itemize}
        \item $\forall\ j \in [3,n-3]. \mbox{ if }(\atr_{j+1},\atr_{j+2}) \in \sametro\mbox{ and }(\atr_{j+2},\atr_{j+3}) \in \mcfo \mbox{ then }(\atr_{j},\atr_{j+1}) \in \msto$.
    \end{itemize}
    then $\aprog$ is robust against \pcc{} relative to \sic{}.
\end{theorem}
In Figure \ref{fig:litmus4CDG}, we have three simple cycles in the graph: 
\begin{itemize}[topsep=3pt]
    \item $(\atrwr{\atr 1}, \atrrd{\atr 4}) \in  \mrfo$ and $(\atrrd{\atr 4}, \atrwr{\atr 1}) \in  \mcfo$,
    \item $(\atrwr{\atr 2}, \atrrd{\atr 3}) \in  \mrfo$ and $(\atrrd{\atr 3}, \atrwr{\atr 2}) \in  \mcfo$,
    \item $(\atrwr{\atr 1}, \atrwr{\atr 2}) \in  \po$, $(\atrwr{\atr 2}, \atrrd{\atr 3}) \in  \mrfo$, $(\atrrd{\atr 3}, \atrrd{\atr 4}) \in  \po$, and $(\atrrd{\atr 4}, \atrwr{\atr 1}) \in  \mcfo$.
\end{itemize}
Notice that none of the cycles satisfies the properties in Theorems \ref{them:MovRobCcPc} and \ref{them:MovRobPcSi}. 
Therefore, $\mathsf{MP}$ is robust against \ccc{} relative to \pcc{} and against \pcc{} relative to \sic{}.

\begin{remark} \label{rem:comgraph}
For programs that contain an unbounded number of processes, an unbounded number of instantiations of a fixed number of process ``templates'', or unbounded loops with bodies that contain entire transactions, a sound robustness check consists in applying Theorem~\ref{them:MovRobCcPc} and Theorem~\ref{them:MovRobPcSi} to (bounded) programs that contain two copies of each process template, and where each loop is unfolded exactly two times. This holds because the mover relations are ``static'', they do not depend on the context in which the transactions execute, and each cycle requiring more than two process instances or more than two loop iterations can be short-circuited to a cycle that exists also in the bounded program. Every outgoing edge from a third instance/iteration can also be taken from the second instance/iteration. Two copies/iterations are necessary in order to discover cycles between instances of the same transaction (the cycles in Theorem~\ref{them:MovRobCcPc} and Theorem~\ref{them:MovRobPcSi} are simple and cannot contain the same transaction twice). These results extend easily to SQL queries as well because the notion of mover is independent of particular classes of programs or instructions.
\end{remark}

\vspace{-15pt}
\section{Experimental Evaluation}
\label{sec:experiments}

We evaluated our approach for checking robustness on 7 applications extracted from the literature on databases and distributed systems, and an application Betting designed by ourselves. Two applications were extracted from the OLTP-Bench benchmark~\cite{DBLP:journals/pvldb/DifallahPCC13}: a vote recording application (Vote) and a consumer review application (Epinions). Three applications were obtained from Github projects (used also in~\cite{DBLP:conf/cav/BeillahiBE19,DBLP:conf/pldi/BrutschyD0V18}): a distributed lock application for the Cassandra database (CassandraLock~\cite{lock}), an application for recording trade activities (SimpleCurrencyExchange~\cite{trade}), and a micro social media application (Twitter~\cite{twitter}). The last two applications are a movie ticketing application (FusionTicket)~\cite{DBLP:conf/cloud/HoltBZPOC16}, and a user subscription application inspired by the Twitter application (Subscription).  Each application consists of a set of SQL transactions that can be called an arbitrary number of times from an arbitrary number of processes. For instance, Subscription provides an AddUser transaction for adding a new user with a given username and password, and a RemoveUser transaction for removing an existing user. (The examples in Figure~\ref{fig:exmapleprograms} are particular variations of FusionTicket, Twitter, and Betting.) 
We considered five variations of the robustness problem: the three robustness problems we studied in this paper along with robustness against \sic{} relative to \serc{} and against \ccc{} relative to \serc{}. The artifacts are available in a GitHub repository \cite{githubExp}.

\begin{table}[!ht]
    \vspace{-10mm}
    \caption{Results of the experiments.
    The columns titled \textsf{X}-\textsf{Y} stand for the result of applications robustness against $\textsf{X}$ relative to $\textsf{Y}$.}
    \centering
    \begin{tabular}{|c|c|c|c|c|c|c|}
    \hline
    Application & Transactions & \multicolumn{5}{c|}{Robustness}\\
     & & \ccc{}-\pcc{} & \pcc{}-\sic{} & \ccc{}-\sic{}  & \sic{}-\serc{} & \ccc{}-\serc{}\\
    \hline
    Betting & 2 & yes & yes & yes & yes & yes\\
    \hline
    CassandraLock & 3 & yes & yes & yes & yes & yes\\
    \hline
    Epinions &  8 & no & yes & no & yes & no\\ 
    \hline
    FusionTicket & 3 & no & no & no & yes & no\\
    \hline
    SimpleCurrencyExchange & 4 & yes & yes & yes & yes & yes\\ 
    \hline
    Subscription & 2 & yes & no & no & yes & no \\
    \hline
    Twitter & 3 & no & no & no & yes & no\\
    \hline
    Vote & 1 & yes & yes & yes & no & no \\
    \hline
    \hline
    \end{tabular}
    \label{tab:exper}
    \vspace{-5mm}
\end{table}

In the first part of the experiments, we check for robustness violations in bounded-size executions of a given application.
For each application, we have constructed a client program with a fixed number of processes (2) and a fixed number of transactions of the corresponding application (at most 2 transactions per process). 
For each program and pair of consistency models, we check for robustness violations using the reductions to reachability under \serc{} presented in Section~\ref{sec:CCPCrobustness} and Section~\ref{sec:robustness} in the case of pairs of weak consistency models, and the reductions in~\cite{DBLP:conf/cav/BeillahiBE19,DBLP:conf/concur/BeillahiBE19} when checking for robustness relative to \serc{}.
We check for reachability (assertion violations) using the Boogie program verifier \cite{DBLP:conf/fmco/BarnettCDJL05}. We model tables as unbounded maps in Boogie and SQL queries as first-order formulas over these maps (that may contain existential or universal quantifiers). To model the uniqueness of primary keys we use Boogie linear types. 

Table~\ref{tab:exper} reports the results of this experiment (cells filled with ``no'')\footnote{The Twitter client in Table~\ref{tab:exper}, which is not \pcc{} vs \ccc{} robust, is different from the one described in Section~\ref{sec:overview}. This client program consists of two processes, each executing FollowUser and AddTweet.}. Five applications are not robust against at least one of the semantics relative to some other stronger semantics. 
The runtimes (wall-clock times) for the robustness checks are all under one second, and the memory consumption is around $50$ Mega\-bytes. Concerning scalability, the reductions to reachability presented in Section~\ref{sec:CCPCrobustness} and Section~\ref{sec:robustness} show that checking robustness is as hard as checking reachability (the size of the instrumented program is only linear in the size of the original program). Therefore, checking robustness will also suffer from the classic state explosion problem when increasing the number of processes. On the other hand, increasing the number of transactions in a process does not seem to introduce a large overhead. Increasing the number of transactions per process in the clients of Epinions, FusionTicket, and SimpleCurrencyExchange from 2 to 5 introduces a running time overhead of at most 25\%.


All the robustness violations we report correspond to violations of the intended specifications. For instance: (1) the robustness violation of Epinions against \ccc{} relative to \pcc{} allows two users to update their ratings for a given product and then when each user queries the overall rating of this product they do not observe the latest rating that was given by the other user, (2) the robustness violation of Subscription against \pcc{} relative to \sic{} allows two users to register new accounts with the same identifier, and (3) the robustness violation of Vote against \sic{} relative to \serc{} allows the same user to vote twice. The specification violation in Twitter was reported in \cite{DBLP:conf/pldi/BrutschyD0V18}. However, it was reported as violation of a different robustness property (\ccc{} relative to \serc{}) while our work shows that the violation persists when replacing a weak consistency model (e.g., \sic{}) with a weaker one (e.g. \ccc{}). This implies that this specification violation is not present under \sic{} (since it appears in the difference between \ccc{} and \sic{} behaviors), which cannot be deduced from previous work.

In the second part of the experiments, we used the technique described in Section \ref{sec:commutativitygraph}, based on commutativity dependency graphs, to prove robustness. For each application (set of transactions) we considered a program that for each ordered pair of (possibly identical) transactions in the application, contains two processes executing that pair of transactions. Following Remark~\ref{rem:comgraph}, the robustness of such a program implies the robustness of a \emph{most general client} of the application that executes each transaction an arbitrary number of times and from an arbitrary number of processes. 
We focused on the cases where we could not find robustness violations in the first part.
To build the ``non-mover'' relations $\mrfo$, $\msto$, and $\mcfo$ for the commutativity dependency graph, we use the left/right mover check provided by the CIVL verifier~\cite{DBLP:conf/cav/HawblitzelPQT15}. 
The results are reported in Table~\ref{tab:exper}, the cells filled with ``yes''. 
We showed that the three applications Betting, CassandraLock and SimpleCurrencyExchange are robust against any semantics relative 
to some other stronger semantics. 
As mentioned earlier, all these robustness results are established for arbitrarily large executions and clients with an arbitrary number of processes. For instance, the robustness of SimpleCurrencyExchange ensures that when the exchange market owner observes a trade registered by a user, they observe also all the other trades that were done by this user in the past. 

In conclusion, our experiments show that the robustness checking techniques we present are effective in proving or disproving robustness of concrete applications. Moreover, it shows that the robustness property for different combinations of consistency models is a relevant design principle, that can help in choosing the right consistency model for realistic applications, i.e., navigating the tradeoff between consistency and performance (in general, weakening the consistency leads to better performance).


\vspace{-5pt}
\section{Related Work}
\label{sec:related}
\vspace{-5pt}
 
The consistency models in this paper were studied in several recent works~\cite{DBLP:conf/popl/BurckhardtGYZ14,DBLP:journals/ftpl/Burckhardt14,DBLP:conf/concur/Cerone0G15,DBLP:conf/ppopp/PerrinMJ16,DBLP:conf/popl/BouajjaniEGH17,DBLP:conf/vmcai/RaadLV19,DBLP:journals/pacmpl/BiswasE19}. Most of them focused on their operational and axiomatic formalizations. 
The formal definitions we use in this paper are based on those given in~\cite{DBLP:conf/concur/Cerone0G15,DBLP:conf/popl/BouajjaniEGH17}.  Biswas and Enea~\cite{DBLP:journals/pacmpl/BiswasE19} shows that checking whether an execution is \ccc{} is polynomial time while checking whether it is \pcc{} or \sic{} is NP-complete.

The robustness problem we study in this paper has been investigated in the context of weak memory models, but only relative to sequential consistency, against Release/Aquire (RA), TSO and Power~\cite{DBLP:conf/pldi/LahavM19,DBLP:conf/icalp/BouajjaniMM11, DBLP:conf/esop/BouajjaniDM13,DBLP:conf/icalp/DerevenetcM14}. Checking robustness against \ccc{}  and \sic{} relative to \serc{} has been investigated in~\cite{DBLP:conf/cav/BeillahiBE19,DBLP:conf/concur/BeillahiBE19}. 
In this work, we study the robustness problem between two weak consistency models, which poses different non-trivial challenges.
In particular, previous work proposed reductions to reachability under sequential consistency (or \serc{}) that relied on a concept of minimal robustness violations (w.r.t. an operational semantics), which does not apply in our case. 
The relationship between \pcc{} and \serc{} is similar in spirit to the one given by Biswas and Enea~\cite{DBLP:journals/pacmpl/BiswasE19} in the context of checking whether an execution is \pcc{}. However, that relationship was proven in the context of a ``weaker'' notion of trace (containing only program order and read-from), and it does not extend to our notion of trace. For instance, that result does not imply preserving $\sto$ dependencies which is crucial in our case.

Some works describe various over- or under-approximate analyses for checking robustness relative to \serc{}. The works in~\cite{DBLP:conf/concur/0002G16,DBLP:conf/popl/BrutschyD0V17,DBLP:conf/pldi/BrutschyD0V18,DBLP:journals/jacm/CeroneG18,DBLP:conf/concur/NagarJ18} propose static analysis techniques based on computing an abstraction of the set of computations, which is used for proving robustness.
In particular, \cite{DBLP:conf/pldi/BrutschyD0V18,DBLP:conf/concur/NagarJ18} encode program executions under the weak consistency model 
using FOL formulas to describe the dependency relations between actions in the executions.
These approaches may return false alarms due to the abstractions they consider in their encoding.
Note that in this paper, we prove a strengthening of the results of \cite{DBLP:conf/concur/0002G16} with regard to the shape of happens before cycles allowed under \pcc{}.

An alternative to {\em trace-based} robustness, is {\em state-based} robustness which requires that a program is robust if the sets of reachable states under two semantics coincide. 
While state-robustness is the necessary and sufficient concept for preserving state-invariants, its verification, which amounts in computing the set of reachable states under the weak semantics models is in general a hard problem.
The decidability and the complexity of this problem has been investigated in the context of relaxed memory models such as TSO and Power, and it has been shown that it is either decidable but highly complex (non-primitive recursive), or undecidable \cite{DBLP:conf/popl/AtigBBM10,DBLP:conf/esop/AtigBBM12}. 
Automatic procedures for approximate reachability/invariant checking have been proposed using either abstractions or bounded analyses, e.g., \cite{DBLP:conf/cav/AtigBP11,DBLP:conf/cav/AlglaveKT13,DBLP:journals/cl/DanMVY17,DBLP:conf/tacas/AbdullaABN17}. Proof methods have also been developed for verifying invariants in the context of weakly consistent models such as \cite{DBLP:conf/icalp/LahavV15,DBLP:conf/popl/GotsmanYFNS16,DBLP:conf/eurosys/NajafzadehGYFS16,DBLP:conf/popl/AlglaveC17}. These methods, however, do not provide decision procedures.

    \bibliographystyle{splncs04}
    \bibliography{draft}


\vfill

{\small\medskip\noindent{\bf Open Access} This chapter is licensed under the terms of the Creative Commons\break Attribution 4.0 International License (\url{http://creativecommons.org/licenses/by/4.0/}), which permits use, sharing, adaptation, distribution and reproduction in any medium or format, as long as you give appropriate credit to the original author(s) and the source, provide a link to the Creative Commons license and indicate if changes were made.}

{\small \spaceskip .28em plus .1em minus .1em The images or other third party material in this chapter are included in the chapter's Creative Commons license, unless indicated otherwise in a credit line to the material.~If material is not included in the chapter's Creative Commons license and your intended\break use is not permitted by statutory regulation or exceeds the permitted use, you will need to obtain permission directly from the copyright holder.}

\medskip\noindent\includegraphics{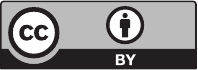}

\appendix
\newpage
\section{Proofs for Section \ref{sec:CCPCrobustness}}
\label{sec:CCPCrobustnessProofs}

\begin{proof}[Proof of Lemma \ref{lem:Transform}]
We start with the case $\textsf{X} = \ccc$.  We first show that $\atrace_\pcinstr$ satisfies $\axpoco$ and $\axcoarb$. For $\axpoco$, let $\atr_{1}' \in \{\atrrd{\atr_{1}},\atrwr{\atr_{1}}\}$ and $\atr_{2}' \in \{\atrrd{\atr_{2}},\atrwr{\atr_{2}}\}$, such that $(\atr_{1}',\atr_{2}') \in (\po_\pcinstr\cup\rfo_\pcinstr)^{+}$. By the definition of $\viso_\pcinstr$, we have that either $(\atr_{1}'=\atrrd{\atr_{1}},\atr_{2}'=\atrwr{\atr_{2}}) \in \po_\pcinstr$ and $\atr_1 = \atr_2$ or $(\atr_{1},\atr_2) \in (\po\cup\rfo)^{+}$, which implies that $(\atr_{1},\atr_2) \in \viso$. In both cases we obtain that $(\atr_{1}',\atr_{2}') \in \viso_\pcinstr$. The axiom $\axpoco$ can be proved in a similar way.

    Next, we show that $\atrace_\pcinstr$ satisfies the property $\axretval$.  Let $\atr$ be a transaction in $\atrace$ that contains a read event $\readact(\atr,\anaddr,\aval)$. 
Let $\atr_0$ be the transaction in $\atrace$ such that 
$$\atr_0 = Max_{\arbo}(\{\atr' \in \atrace\ |\ (\atr',\atr) \in \viso \wedge \exists\ \writeact(\atr',\anaddr,\cdot)\in\amap(\atr')\}).$$ 
The read value $\aval$ must have been written by $\atr_0$ since $\atrace$ satisfies $\axretval$. Thus, the read $\readact(\atr,\anaddr,\aval)$ in $\atrrd{\atr}$ of $\atrace_\pcinstr$ must return the value written by $\atrwr{\atr_0}$. 
From the definitions of $\viso_\pcinstr$ and $\arbo_\pcinstr$, we get  
$$\atrwr{\atr_1}\in\{\atrwr{\atr'} \in \atrace_\pcinstr\ |\ (\atrwr{\atr'},\atrrd{\atr}) \in \viso_\pcinstr \wedge \exists\ \writeact(\atrwr{\atr'},\anaddr,\cdot)\in\amap(\atrwr{\atr'})\}$$ 
iff 
$$\atr_1 \in \{\atr' \in \atrace\ |\ (\atr',\atr) \in \viso \wedge \exists\ \writeact(\atr',\anaddr,\cdot)\in\amap(\atr')\}$$ 
because $(\atrwr{\atr_{1}},\atrrd{\atr_{2}}) \in \viso_\pcinstr$ implies $(\atr_{1},\atr_{2}) \in \viso$. Since $(\atrwr{\atr_{1}},\atrwr{\atr_{2}}) \in \arbo_\pcinstr$ implies $(\atr_{1},\atr_{2}) \in \arbo$, we also obtain that 
$$\atrwr{\atr_0} = Max_{\arbo_\pcinstr}(\{\atrwr{\atr'} \in \atrace_\pcinstr\ |\ (\atrwr{\atr'},\atrrd{\atr}) \in \viso_\pcinstr \wedge \exists\ \writeact(\atrwr{\atr'},\anaddr,\cdot)\in\amap(\atrwr{\atr'})\})$$ 
and since the read $\readact(\atr,\anaddr,\aval)$ in $\atrrd{\atr}$ of $\atrace_\pcinstr$ returns the value written by $\atrwr{\atr_0}$, $\atrace_\pcinstr$ satisfies $\axretval$. 

For the case $\textsf{X} = \pcc$, we show that $\atrace_\pcinstr$ satisfies the property $\axprefix$ (the other axioms are proved as in the case of $\ccc$). 
    Suppose we have $(\atr_{1}',\atr_{2}') \in \arbo_\pcinstr$ and $(\atr_{2}',\atr_{3}') \in \viso_\pcinstr$ where $\atr_{1}' \in \{\atrrd{\atr_{1}},\atrwr{\atr_{1}}\}$, $\atr_{2}' \in \{\atrrd{\atr_{2}},\atrwr{\atr_{2}}\}$, and $\atr_{3}' \in \{\atrrd{\atr_{3}},\atrwr{\atr_{3}}\}$. The are five cases to be discussed:
    \begin{enumerate} 
        \item $(\atr_{1}'=\atrrd{\atr_{1}},\atr_{2}'=\atrwr{\atr_{2}}) \in \po_\pcinstr$ and $\atr_1 = \atr_2$ and $(\atr_{2},\atr_{3}) \in \viso$, 
        \item $(\atr_{1},\atr_{2}) \in \arbo$ and $(\atr_{2},\atr_{3}) \in \viso$,  
        \item $(\atr_{1},\atr_{2}) \in \arbo$ and $(\atr_{2}'=\atrrd{\atr_{2}},\atr_{3}'=\atrwr{\atr_{3}}) \in \po_\pcinstr$ and $\atr_2 =  \atr_3$, 
        \item $(\atr_{1}'=\atrrd{\atr_{1}},\atr_{2}'=\atrwr{\atr_{2}}) \in \po_\pcinstr$ and $\atr_1 = \atr_2$ and $(\atr_{2},\atr_{3}) \in \arbo$ and $\atr_{3}'= \atrwr{\atr_{3}}$, 
        \item $(\atr_{1},\atr_{2}) \in \arbo$ and $(\atr_{2},\atr_{3}) \in \arbo$ and $\atr_{3}'= \atrwr{\atr_{3}}$. 
    \end{enumerate}
    Cases (a) and (b) imply that $(\atr_{1},\atr_{3}) \in \viso$ since $\arbo;\viso \subset \arbo$, which implies that $(\atr_{1}',\atr_{3}') \in \viso_\pcinstr$. Cases (c), (d), and (e) imply that $(\atr_{1},\atr_{3}) \in \arbo$ and $\atr_{3}'= \atrwr{\atr_{3}}$ then 
   we get that $(\atrwr{\atr_{1}},\atrwr{\atr_{3}}) \in \viso_\pcinstr$ and $\atr_{3}'= \atrwr{\atr_{3}}$ which means that 
   $(\atr_{1}',\atr_{3}') \in \viso_\pcinstr$. 

Note that the rule $(\atrwr{\atr_{1}},\atrwr{\atr_{2}}) \in \viso_\pcinstr$ if $(\atr_{1},\atr_{2}) \in \arbo$ cannot change the fact that $$\atrwr{\atr_1}\in\{\atrwr{\atr'} \in \atrace_\pcinstr\ |\ (\atrwr{\atr'},\atrrd{\atr}) \in \viso_\pcinstr \wedge \exists\ \writeact(\atrwr{\atr'},\anaddr,\cdot)\in\amap(\atrwr{\atr'})\}$$ iff $$\atr_1 \in \{\atr' \in \atrace\ |\ (\atr',\atr) \in \viso \wedge \exists\ \writeact(\atr',\anaddr,\cdot)\in\amap(\atr')\}$$ Thus, the proof of $\axretval$ follows as in the previous case.

For the case $\textsf{X} = \sic$, we show that $\atrace_\pcinstr$ satisfies $\axconflict$. If $(\atrwr{\atr_{1}},\atrwr{\atr_{2}}) \in \sto_\pcinstr$, then $(\atr_{1},\atr_{2}) \in \sto \subset \viso$, which implies that $(\atrwr{\atr_{1}},\atrrd{\atr_{2}}) \in \viso_\pcinstr$. Therefore, $(\atrwr{\atr_{1}},\atrwr{\atr_{2}}) \in \viso_\pcinstr$, which concludes the proof. The axiom $\axretval$ can be proved as in the previous cases.
\end{proof}

\begin{proof}[Proof of Lemma \ref{lem:cycles}]
($\Rightarrow$) Let $\atrace$ be a trace under \ccc{}. From $\axpoco$ and $\axcoarb$ we get that $\arbo_{0}^{+} \subset \arbo$, and $\arbo_{0}^{+}$ is acyclic because $\arbo$ is total order.  Assume by contradiction that $\viso_{0}^{+};\cfo$ is cyclic which implies that $\viso;\cfo$ is cyclic since $\viso_{0}^{+} \subset \viso$, which means that there exist $\atr_1$ and $\atr_2$ such that $(\atr_1, \atr_2) \in \viso$ and $(\atr_2, \atr_1) \in \cfo$. 
$(\atr_2, \atr_1) \in \cfo$ implies that there exists $\atr_3$ such that $(\atr_3, \atr_1) \in \sto$ and $(\atr_3, \atr_2) \in \rfo$. 
Based on the definition of $\axretval$, $\atr_3$ has two possible instances:
\begin{itemize}
    \item $\atr_3$ corresponds to the "fictional" transaction that wrote the initial values which cannot be the case when $(\atr_1, \atr_2) \in \viso$ and $\atr_1$ writes to the same variable that $\atr_2$ reads from,
    \item $\atr_3$ is the last transaction that occurs before $\atr_2$ that writes the value read by $\atr_2$, which means that 
    $(\atr_1,\atr_3) \in \arbo$ which contradicts the fact that $(\atr_3, \atr_1) \in \sto$ since $\sto \subset \arbo$.
\end{itemize}  

($\Leftarrow$) Let $\atrace$ be a trace such that $\arbo_{0}^{+}$ and $\viso_{0}^{+};\cfo$ are acyclic. Then, we define the relations $\viso$ and $\arbo$ such that $\viso = \viso_{0}^{+}$ and $\arbo$ is any total order that includes $\arbo_{0}^{+}$. Then, we obtain that $(\viso \cup \sto)^{+} \subset \arbo$ and $\viso;\cfo$ is acyclic. Thus, $\atrace$ satisfies the properties $\axpoco$ and $\axcoarb$. Next, we will show that $\atrace$ satisfies $\axretval$. 
Let $\atr$ be a transaction in $\atrace$ that contains a read event $\readact(\atr,\anaddr,\aval)$. 
Let $\atr_0$ be transaction in $\atrace$ such that $$\atr_0 = Max_{\arbo}(\{\atr' \in \atrace\ |\ (\atr',\atr) \in \viso \wedge \exists\ \writeact(\atr',\anaddr,\cdot)\in\amap(\atr')\})$$ then the read must return a value written by $\atr_0$. 
Assume by contradiction that there exists some other transaction $\atr_1 \neq \atr_0$ such that $(\atr_1,\atr) \in \rfo$. 
Then, we get that $(\atr_1,\atr_0) \in \arbo$ and both write to $\anaddr$, therefore, $(\atr_1,\atr_0) \in \sto$ since $\sto \subset \arbo$. Combining $(\atr_1,\atr) \in \rfo$ and $(\atr_1,\atr_0) \in \sto$ we obtain $(\atr,\atr_0) \in \cfo$ and since 
$(\atr_0,\atr) \in \viso$ then we obtain that $(\atr,\atr) \in \viso;\cfo$ which contradicts the fact that $\viso;\cfo$ is acyclic.
Therefore, the read value was written by $\atr_0$ and $\atrace$ satisfies $\axretval$.
%
%
%
%
\end{proof}

\begin{proof}[Proof of Lemma~\ref{lem:CcCc}]
    The only-if direction follows from Lemma \ref{lem:Transform}. For the if direction: consider a trace $\atrace_\pcinstr$ which is \ccc{}. We prove by contradiction that $\atrace$ must be \ccc{} as well. 
    Assume that $\atrace$ is not \ccc{} then it must contain a cycle in either $\arbo_{0}^{+}$ or $\viso_{0}^{+};\cfo$ (based on Lemma \ref{lem:cycles}). In the rest of the proof when we mention a cycle we implicitly refer to a cycle in either $\arbo_{0}^{+}$ or $\viso_{0}^{+};\cfo$. 
    
    Splitting every transaction $\atr \in \atrace$ in a trace to a pair of transactions $\atrrd{\atr}$ and $\atrwr{\atr}$ that occur in this order might not maintain a cycle of $\atrace$. However, we prove that this is not possible and our splitting conserves the cycle. 
    Assume we have a vertex $\atr$ as part of the cycle. We show that $\atr$ can be split into two transactions
     $\atrrd{\atr}$ and $\atrwr{\atr}$ while maintaining the cycle. 
    Note that $\atr$ is part of a cycle iff either 
    \begin{enumerate}
        \item  $(\atr_{1},\atr) \in \arbo_{0}$ and $(\atr,\atr_{2})\in \arbo_{0}$ or
        \item  $(\atr_{1},\atr) \in \viso_{0}$ and $(\atr,\atr_{2})\in \viso_{0}$ or
        \item  $(\atr_{1},\atr) \in \viso_{0}$ and $(\atr,\atr_{2})\in \cfo$ or 
        \item $(\atr_{1},\atr) \in \cfo$ and $(\atr,\atr_{2})\in \viso_{0}$ 
    \end{enumerate}
    where $\atr_{1}$ and $\atr_{2}$ might refer to the same transaction. 
    Thus, by splitting $\atr$ to $\atrrd{\atr}$ and $\atrwr{\atr}$, the above four cases imply that:
    \begin{enumerate}
        \item if $(\atr_{1},\atr) \in \viso_{0}$ and $(\atr,\atr_{2})\in \arbo_{0}$ then 
        $(\atr_{1}',\atrrd{\atr}) \in (\po_\pcinstr \cup \rfo_\pcinstr)$ and $(\atrwr{\atr},\atr_{2}')\in (\po_\pcinstr \cup \rfo_\pcinstr \cup \sto_\pcinstr)$ where $\atr_{1}' \in \{\atrrd{\atr_{1}},\atrwr{\atr_{1}}\}$ and $\atr_{2}' \in \{\atrrd{\atr_{2}},\atrwr{\atr_{2}}\}$. This maintains the vertices $\atr_{1}'$ and $\atr_{2}'$ connected in the cycle formed by the dependency relations of $\atrace_\pcinstr$ since $(\atrrd{\atr},\atrwr{\atr}) \in \po_\pcinstr$;
        \item if $(\atr_{1},\atr) \in \sto$ and $(\atr,\atr_{2})\in \arbo_{0}$ then 
        $(\atr_{1}',\atrwr{\atr}) \in \sto_\pcinstr$ and $(\atrwr{\atr},\atr_{2}')\in (\po_\pcinstr \cup \rfo_\pcinstr \cup \sto_\pcinstr)$ which maintains the vertices $\atr_{1}'$ and $\atr_{2}'$ connected in the cycle formed by the dependency relations of $\atrace_\pcinstr$;
        \item $(\atr_{1},\atr) \in \viso_{0}$ and $(\atr_{2},\atr) \in \cfo$ then $(\atr_{1}',\atrrd{\atr}) \in (\po_\pcinstr \cup \rfo_\pcinstr)$ and $(\atrrd{\atr},\atr_{2}')\in \cfo_\pcinstr$ maintains the vertices $\atr_{1}'$ and $\atr_{2}'$ connected in the cycle formed by the dependency relations of $\atrace_\pcinstr$;
        \item $(\atr_{1},\atr) \in \cfo$ and $(\atr_{2},\atr) \in \viso_{0}$ then $(\atr_{1}',\atrwr{\atr}) \in \cfo_\pcinstr$ and $(\atrwr{\atr},\atr_{2}')\in (\po_\pcinstr \cup \rfo_\pcinstr)$ which maintains the vertices $\atr_{1}'$ and $\atr_{2}'$ connected in the cycle formed by the dependency relations of $\atrace_\pcinstr$ as well. 
    \end{enumerate}  
    Therefore, doing the splitting creates a cycle in either $(\po_\pcinstr \cup \rfo_\pcinstr \cup \sto_\pcinstr)^{+}$ or $(\po_\pcinstr \cup \rfo_\pcinstr)^{+};\cfo_\pcinstr$ which implies that $\atrace_\pcinstr$ is not \ccc{}, a contradiction.
    \end{proof}

\begin{proof}[Proof of Lemma \ref{lem:PcSer}]
($\Leftarrow$) Assume that $\atrace_\pcinstr$ is \serc{}. We will show that $\atrace$ is \pcc{}. 
    Notice that if $(\atr_{1},\atr_{2}) \in \viso_{0}^{+}$ then $(\atrwr{\atr_{1}},\atrrd{\atr_{2}}) \in \arbo_\pcinstr$
    which implies that $(\atr_{1},\atr_{2}) \in \viso$. Similarly, if $(\atr_{1},\atr_{2}) \in \arbo_{0}^{+}$ then $(\atrwr{\atr_{1}},\atrwr{\atr_{2}}) \in \arbo_\pcinstr$ or $(\atrwr{\atr_{1}},\atrrd{\atr_{2}}) \in \arbo_\pcinstr$ which implies that $(\atrwr{\atr_{1}},\atrwr{\atr_{2}}) \in \arbo_\pcinstr$ which in both cases implies that $(\atr_{1},\atr_{2}) \in \arbo$. Thus, $\atrace$ satisfies the properties $\axpoco$ and $\axcoarb$. 
    
    Now assume that $(\atr_{1},\atr_{2})\in \arbo$ and $(\atr_{2},\atr_{3})\in \viso$. We show that $(\atr_{1},\atr_{3})\in \viso$. 
    The assumption implies that $(\atrwr{\atr_{1}},\atrwr{\atr_{2}}) \in \arbo_\pcinstr$ and 
   $(\atrwr{\atr_{2}},\atrrd{\atr_{3}}) \in \arbo_\pcinstr$, which means that $(\atrwr{\atr_{1}},\atrrd{\atr_{3}}) \in \arbo_\pcinstr$. Therefore, $(\atr_{1},\atr_{3}) \in \viso$ and $\atrace$ satisfies the property $\axconflict$. 
   
   Concerning $\axretval$, let $\atr$ be a transaction in $\atrace$ that contains a read event $\readact(\atr,\anaddr,\aval)$. 
   Let $\atr_0$ be transaction in $\atrace$ such that 
   $$\atr_0 = Max_{\arbo}(\{\atr' \in \atrace\ |\ (\atr',\atr) \in \viso \wedge \exists\ \writeact(\atr',\anaddr,\cdot)\in\amap(\atr')\}).$$ 
   We show that the read must return a value written by $\atr_0$. 
   The definitions of $\viso$ and $\arbo$ imply that 
   $$\atrwr{\atr_1}\in\{\atrwr{\atr'} \in \atrace_\pcinstr\ |\ (\atrwr{\atr'},\atrrd{\atr}) \in \arbo_\pcinstr \wedge \exists\ \writeact(\atrwr{\atr'},\anaddr,\cdot)\in\amap(\atrwr{\atr'})\}$$ 
   iff 
   $$\atr_1 \in \{\atr' \in \atrace\ |\ (\atr',\atr) \in \viso \wedge \exists\ \writeact(\atr',\anaddr,\cdot)\in\amap(\atr')\}$$ 
   because $(\atrwr{\atr_{1}},\atrrd{\atr_{2}}) \in \arbo_\pcinstr$ implies $(\atr_{1},\atr_{2}) \in \viso$. 
   Then, we obtain that 
   $$\atrwr{\atr_0} = Max_{\arbo_\pcinstr}(\{\atrwr{\atr'} \in \atrace_\pcinstr\ |\ (\atrwr{\atr'},\atrrd{\atr}) \in \arbo_\pcinstr \wedge \exists\ \writeact(\atrwr{\atr'},\anaddr,\cdot)\in\amap(\atrwr{\atr'})\})$$ 
   and since $\atrace_\pcinstr$ is \serc{} we know that the read must return the value written by $\atrwr{\atr_0}$. Thus, the read returns the value written by $\atr_0$, which implies that $\atrace$ satisfies $\axretval$ holds. Therefore, $\atrace$ is \pcc{}. 

($\Rightarrow$) Assume that $\atrace$ is \pcc{}. We show that $\atrace_\pcinstr$ is \serc{}. 
   Since $\atrace_\pcinstr$ is the result of splitting transactions, a cycle in its dependency relations can only originate from a cycle in $\atrace$. Therefore, it is sufficient to show that any happens-before cycle in $\atrace$ is broken in $\atrace_\pcinstr$. 
    From Lemma \ref{lem:pccycles}, we have that $\atrace$ either does not admit a happens-before cycle or any (simple) happens-before cycle in $\atrace$ must have either two successive $\cfo$ dependencies or a $\sto$ dependency followed by a $\cfo$ dependency. 
    If $\atrace$ does not admit a happens-before cycle then it is \serc{}, and $\atrace_\pcinstr$  is trivially \serc{} (since splitting transactions cannot introduce new cycles). 

\scalebox{0.67}
{
\begin{tikzpicture}
 \node[shape=rectangle ,draw=none,font=\large] (A0) at (0,0)  [] {$\atr_1$ };
 \node[shape=rectangle ,draw=none,font=\large] (A1) at (2,0)  [] {$\atr_2$};
 \node[shape=rectangle ,draw=none,font=\large] (B1) at (4,0)  [] {$\atr_3$};
 \node[shape=rectangle ,draw=none,font=\large] (B2) at (5,0)  [] {$\Longrightarrow$};
 \node[shape=rectangle ,draw=none,font=\large] (C1) at (6,0)  [] {$\atrrd{\atr_{1}}$};
 \node[shape=rectangle ,draw=none,font=\large] (D1) at (8,0)  [] {$\atrwr{\atr_{1}}$};
 \node[shape=rectangle ,draw=none,font=\large] (D2) at (10,0)  [] {$\atrrd{\atr_{2}}$};
 \node[shape=rectangle ,draw=none,font=\large] (D0) at (12,0)  [] {$\atrwr{\atr_{2}}$};
 \node[shape=rectangle ,draw=none,font=\large] (E0) at (14,0)  [] {$\atrrd{\atr_{3}}$ };
 \node[shape=rectangle ,draw=none,font=\large] (E1) at (16,0)  [] {$\atrwr{\atr_{3}}$};

  \begin{scope}[ every edge/.style={draw=black,very thick}]
  \path [->] (A0) edge[] node [above,font=\small] {$\sto \cup \cfo$}  (A1);
  \path [->] (A1) edge[] node [above,font=\small] {$\cfo$}  (B1);
  \path [->] (B1) edge[bend left] node [above,font=\small] {$\hbo$}  (A0);

  \path [->] (C1) edge[] node [above,font=\small] {$\po_\pcinstr$}  (D1);
  \path [->] (D1) edge[bend left] node [below,font=\small] {$\sto_\pcinstr$}  (D0);
  \path [->] (C1) edge[bend right] node [above,font=\small] {$\cfo_\pcinstr$}  (D0);
  \path [->] (D2) edge[] node [above,font=\small] {$\po_\pcinstr$}  (D0);

  \path [->] (D2) edge[bend right] node [above,font=\small] {$\cfo_\pcinstr$}  (E1);
  \path [->] (E0) edge[] node [above,font=\small] {$\po_\pcinstr$}  (E1);
  \end{scope}
\end{tikzpicture}}

    Otherwise, if $\atrace$ admits a happens-before cycle like above, then $\atrace$ must contain three transactions $\atr_{1}$, $\atr_{2}$, and $\atr_{3}$ such that $(\atr_{1},\atr_{2}) \in \sto \cup \cfo$, $(\atr_{2},\atr_{3}) \in \cfo$, and $(\atr_{3},\atr_{1}) \in \hbo$ (like in the picture above). 
    Then, by splitting transactions we obtain that $(\atrwr{\atr_{1}},\atrwr{\atr_{2}}) \in \sto_\pcinstr$ or $(\atrrd{\atr_{1}},\atrwr{\atr_{2}}) \in \cfo_\pcinstr$, and $(\atrrd{\atr_{2}},\atrwr{\atr_{3}}) \in \cfo_\pcinstr$. 
    Since, we have $(\atrrd{\atr_{2}},\atrwr{\atr_{2}}) \in \po_\pcinstr$ (and not $(\atrwr{\atr_{2}},\atrrd{\atr_{2}}) \in \po_\pcinstr$), this cannot lead to a cycle in $\atrace_\pcinstr$, which concludes the proof that $\atrace_\pcinstr$ is \serc{}
%
\end{proof}

\section{Proofs for Section~\ref{sec:robustness}}
\label{sec:robustnessProofs}

\begin{proof}[Proof of Lemma~\ref{lem:pcsicycles}]
Let $\arbo_{1}$ be a total order that includes $\arbo_{0}^{+}$ and $\arbo_{0}^{+};\cfo;\arbo_{0}^{*}$ ($\arbo_{0}^*$ is the reflexive closure of $\arbo_{0}$). This is well defined because there exists no cycle between tuples in these two relations. Indeed, if $(\atr_{1},\atr_{2}) \in \arbo_{0}^{+}$ and there exist $\atr_{3}$ and $\atr_{4}$ such that $(\atr_{2}, \atr_{3}) \in \arbo_{0}^{+}$, $(\atr_{3},  \atr_{4}) \in \cfo$, and $(\atr_{4}, \atr_{1}) \in \arbo_{0}^{*}$, then we have a cycle in  $\arbo_{0}^{+};\cfo$ that does not contain two successive $\cfo$ dependencies, which contradicts the hypothesis. Also, for every pair of transactions $(\atr_{1},\atr_{2})$ there cannot exist $\atr_{3}$ and $\atr_{4}$ such that 
$$(\atr_{2}, \atr_{3}) \in \arbo_{0}^{+},\ (\atr_{3},  \atr_{4}) \in \cfo\mbox{ and }(\atr_{4}, \atr_{1}) \in \arbo_{0}^{*}$$ and 
$\atr_{3}'$ and $\atr_{4}'$ such that $$(\atr_{1}, \atr_{3}') \in \arbo_{0}^{+},\ (\atr_{3}',  \atr_{4}') \in \cfo\mbox{ and }(\atr_{4}', \atr_{2}) \in \arbo_{0}^{*}$$ 
This will imply a cycle in $\arbo_{0}^{+};\cfo;\arbo_{0}^{+};\cfo$ which again contradicts the hypothesis.
Also, let $\viso_{1}$ be the smallest transitive relation that includes $\arbo_{0}^{+}$ and $\arbo_{1};\arbo_{0}^{+}$. We show that $\viso_{1}$ and $\arbo_{1}$ are causal and arbitration orders of $\atrace$ that satisfy all the axioms of \sic{}. 

$\axpoco$ and $\axcoarb$ hold trivially. Since $\sto \subseteq \viso_{1}$, $\axconflict$ holds as well. 
%
%
%
$\axpc$ holds because $\arbo_{1} ; \viso_{1} = \arbo_{1};(\arbo_{0}^{+} \cup \arbo_{1};\arbo_{0}^{+})^+ = \arbo_{1};\arbo_{0}^{+} \subset \viso_{1}$. 

The axiom $\axretval$ is equivalent to the acyclicity of $\viso_{1};\cfo$ when $\axpoco$ and $\axcoarb$ hold. Assume by contradiction that $\viso_1;\cfo$ is cyclic. 
From the definition of $\viso_1$ and the fact that $\arbo_{1}$ is total order we obtain that either: 
\begin{itemize}
  \item $\arbo_{0}^{+};\cfo$ is cyclic, which implies that there exists a happens-before cycle that does not contain two successive $\cfo$, which contradicts the hypothesis, or
 \item $\arbo_{1};\arbo_{0}^{+};\cfo$ is cyclic, which implies that there exist $\atr_{1}$, $\atr_{2}$, and $\atr_{3}$ such that $(\atr_{2}, \atr_{3}) \in \arbo_{0}^{+}$, $(\atr_{3},  \atr_{1}) \in \cfo$ and $(\atr_{1},\atr_{2}) \in \arbo_{1}$.  This contradicts the fact that $(\atr_{2}, \atr_{3}) \in \arbo_{0}^{+}$ and $(\atr_{3},  \atr_{1}) \in \cfo$ implies $(\atr_{2},\atr_{1}) \in \arbo_{1}$.  
\end{itemize} 
Therefore, $\atrace$ satisfies $\axretval$ for $\viso_1$ and $\arbo_1$, which concludes the proof.
\end{proof}

\medskip
Next, we present an important lemma that characterizes happens before cycles possible under the \pcc{} semantics. 
This is a strengthening of a result in~\cite{DBLP:conf/concur/0002G16} which shows that all happens before cycles under \pcc{} must have two successive dependencies in $\{\cfo,\sto\}$ and at least one $\cfo$. We show that the two successive dependencies cannot be $\cfo$ followed $\sto$, or two successive $\sto$.

\begin{proof}[Proof of Lemma~\ref{lem:pccycles}]
    It was shown in \cite{DBLP:conf/concur/0002G16} that all happens-before cycles under \pcc{} must contain two successive dependencies in $\{\cfo,\sto\}$ and at least one $\cfo$. 
    Assume by contradiction that there exists a cycle with $\cfo$ dependency followed by $\sto$ dependency or two successive $\sto$ dependencies. This cycle must contain at least one additional dependency. Otherwise, the cycle would also have a $\sto$ dependency followed by a $\cfo$ dependency, or it would imply a cycle in $\sto$, which is not possible (since $\sto \subset \arbo$ and $\arbo$ is a total order).
    Then, we get that the dependency just before $\cfo$ is either $\po$ or $\rfo$ (i.e., $\viso_0$) since we cannot have $\cfo$ or $\sto$ followed by $\cfo$. Also, the relation after $\sto$ is either $\po$ or $\rfo$ or $\sto$ (i.e., $\arbo_0$) since we cannot have $\sto$ followed by $\cfo$. Thus, the cycle has the following shape:

\medskip
\scalebox{0.65}
{
\begin{tikzpicture}
 \node[shape=rectangle ,draw=none,font=\large] (A0) at (0,0)  [] {$\atr_1$ };
 \node[shape=rectangle ,draw=none,font=\large] (A1) at (1.3,0)  [] {$\atr_2$};
 \node[shape=rectangle ,draw=none,font=\large] (B1) at (2.6,0)  [] {$\atr_3$};
 \node[shape=rectangle ,draw=none,font=\large] (B2) at (3.9,0)  [] {$\atr_4$};
 \node[shape=rectangle ,draw=none,font=\large] (C0) at (4.5,0)  [] {$\cdots$ };
 \node[shape=rectangle ,draw=none,font=\large] (C1) at (5.1,0)  [] {$\atr_i$};
 \node[shape=rectangle ,draw=none,font=\large] (D1) at (6.4,0)  [] {$\atr_{i+1}$};
 \node[shape=rectangle ,draw=none,font=\large] (D2) at (7.9,0)  [] {$\atr_{i+2}$};
 \node[shape=rectangle ,draw=none,font=\large] (D0) at (9.4,0)  [] {$\atr_{i+3}$};
 \node[shape=rectangle ,draw=none,font=\large] (E0) at (10.2,0)  [] {$\cdots$ };
 \node[shape=rectangle ,draw=none,font=\large] (E1) at (11,0)  [] {$\atr_{n-4}$};
 \node[shape=rectangle ,draw=none,font=\large] (F1) at (12.5,0)  [] {$\atr_{n-3}$};
 \node[shape=rectangle ,draw=none,font=\large] (F2) at (14,0)  [] {$\atr_{n-2}$};
 \node[shape=rectangle ,draw=none,font=\large] (F3) at (15.5,0)  [] {$\atr_{n-1}$};
 \node[shape=rectangle ,draw=none,font=\large] (F4) at (17,0)  [] {$\atr_{n}$};

  \begin{scope}[ every edge/.style={draw=black,very thick}]
  \path [->] (A0) edge[] node [above,font=\small] {$\cfo$}  (A1);
  \path [->] (A1) edge[] node [above,font=\small] {$\sto$}  (B1);
  \path [->] (B1) edge[] node [above,font=\small] {$\arbo_0$}  (B2);

  \path [->] (C1) edge[] node [above,font=\small] {$\viso_0$}  (D1);
  \path [->] (D1) edge[] node [above,font=\small] {$\cfo$}  (D2);
  \path [->] (D2) edge[] node [above,font=\small] {$\sto$}  (D0);

  \path [->] (E1) edge[] node [above,font=\small] {$\viso_0$}  (F1);
  \path [->] (F1) edge[] node [above,font=\small] {$\cfo$}  (F2);
  \path [->] (F2) edge[] node [above,font=\small] {$\sto$}  (F3);
  \path [->] (F3) edge[] node [above,font=\small] {$\arbo_0$} (F4);
  \path [->] (F4) edge[bend left=11] node [above,font=\small] {$\viso_0$} (A0);
  \end{scope}
\end{tikzpicture}}

\medskip
Since $\viso_0;\cfo\subseteq \arbo$ is a consequence of the \pcc{} axioms~\cite{DBLP:conf/concur/CeroneGY17}, we get that $(\atr_{n}, \atr_2) \in \arbo$, $(\atr_{i}, \atr_{i+2}) \in \arbo$ and $(\atr_{n-4}, \atr_{n-2}) \in \arbo$, which allows to ``short-circuit'' the cycle. 
Using the fact that $\sto \subset \arbo$, $\viso_0 \subset \arbo$, and $\arbo_0 \subset \arbo$, and applying the short-circuiting process multiple times, we obtain a cycle in the arbitration order $\arbo$ which contradicts the fact that $\arbo$ is a total order. 
\end{proof}

\begin{proof}[Proof of Theorem~\ref{them:RobCcSi}]
  For the only-if direction: assume that $\aprog$ is robust against \ccc{} relative to \sic{}. Then, 
  the set of traces of $\aprog$ under the two consistency models coincide. 
  Since the set of traces under \sic{} is subset of the one under \pcc{}, then 
  the set of traces under \ccc{} is subset of the one under \pcc{}. 
  This implies that $\aprog$ is robust against \ccc{} relative to \pcc{}. 
  Thus, we obtain that the set of traces of $\aprog$ under the three consistency models coincide. 
  Therefore, $\aprog$ is robust against \pcc{} relative to \sic{} as well.
  
  For the if direction: assume that $\aprog$ is robust against \ccc{} relative to \pcc{} and $\aprog$ is robust against \pcc{} relative to \sic{}. Then, the set of traces of $\aprog$ under the three consistency models coincide. Thus, we obtain that $\aprog$ is robust against \ccc{} relative to \sic{}.     
\end{proof}
    
\section{Proofs for Section~\ref{sec:commutativitygraph}}
\label{sec:commutativitygraphProofs}

\begin{proof}[Proof of Theorem~\ref{them:MovRobCcPc}]
It is enough to show: if $\aprog$ is not robust against \ccc{} relative to \pcc{} then we have a simple cycle in the commutativity dependency graph of $\aprog_\pcinstr$ of the form above.  Assume $\aprog$ is not robust against \ccc{} relative to \pcc{}. 
Then, from Theorem \ref{them:RobCcPc}, we obtain $\aprog_\pcinstr$ is not robust against \ccc{} relative to \serc{}. 
Also it was shown in \cite{DBLP:conf/concur/BeillahiBE19} that if a program is not robust then there must exist a robustness violation trace (\ccc{} relative to \serc{}) $\atrace_\pcinstr$ of the shape $\atrace_\pcinstr = \alpha \cdot \atr_1 \cdot \beta \cdot \atr_i \cdot  \atr_{i+1} \cdot \gamma \cdot \atr_n$ where $(\atr_1,\atr_i) \in (\po \cup \rfo)^{+}$, $(\atr_{i},\atr_{i+1}) \in (\sto \cup \cfo)$, $(\atr_{i+1},\atr_n) \in \hbo$, and  $(\atr_n,\atr_1) \in \cfo$. Note that since transactions in the trace $\atrace_\pcinstr$ can either be read-only or write-only. Then, $(\atr_{i},\atr_{i+1}) \in (\sto \cup \cfo)$ and $(\atr_{n},\atr_1) \in \cfo$ imply that $\atr_1$ and $\atr_{i+1}$ must be a write-only transactions and $\atr_{n}$ must be a read-only transaction. 
Note that we may have $\beta = \gamma = \epsilon$ as the case for the trace of the $\mathsf{SB}$ program given in Figure \ref{fig:litmus1}.

We consider first the general case when $\atr_1 \not\equiv \atr_2$. The other case can be proved in the same way.

Consider the prefix $\atrace_{p}$ of $\atrace_\pcinstr$: $\atrace_{p} = \alpha \cdot \atr_1 \cdot \beta \cdot \atr_{i}$ where $(\atr_1,\atr_{i}) \in (\po \cup \rfo)^{+}$ which is a \serc{} trace of $\aprog_{\pcinstr}$. 
Then, we have a sequence of transactions from $\atr_1$ to $\atr_{i}$ that are related by either $\po$ or $\rfo$. 
In the case two transactions are only related by $\rfo$, then the first transaction is not a right mover because of the second transaction reads from a write in the first transaction. Thus, we can relate the two transactions using the relation $\mrfo$ in the commutativity dependency graph.

Similarly consider the following trace $\atrace_{s}$ extracted from $\atrace_\pcinstr$: $\atrace _{s} = \alpha \cdot  \atr_{i+1} \cdot \gamma \cdot \atr_n$ where $(\atr_{i+1},\atr_n) \in \hbo$ which is a \serc{} trace of $\aprog_{\pcinstr}$. 
Similar to before, we have a sequence of transactions from $\atr_{i+1}$ to $\atr_n$ that are related by either $\po$, $\rfo$, $\sto$, or $\cfo$. 
For any two transactions that are related only by either $\rfo$, $\sto$, or $\cfo$, this implies that the first transaction is not a right mover because of the second transaction and a write-read, write-write, or read-write dependency between the two, respectively. Thus, we can relate the two transactions using either $\mrfo$, $\msto$, or $\mcfo$, respectively. 

Now consider the following trace $\atrace_{1}$ extracted from $\atrace_\pcinstr$: $\atrace_{1} = \alpha \cdot \atr_1 \cdot \beta \cdot \atr_{i} \cdot \atr_{i+1}$ where $(\atr_{i},\atr_{i+1}) \in  (\sto \cup \cfo)$  is a \serc{} trace of $\aprog_{\pcinstr}$.
Because $\atr_{i}$ and $\atr_{i+1}$ are related by either $\sto$ or $\cfo$, then $\atr_{i}$ is not a right mover because of $\atr_{i+1}$ and a write-write or read-write dependency between the two, respectively. 
Thus, we can relate the two transactions using either $\msto$ or $\mcfo$, respectively. 

Finally, consider the following trace $\atrace_{2}$ extracted from $\atrace_\pcinstr$: $\atrace_{2} = \alpha \cdot  \atr_{i+1} \cdot \gamma \cdot \atr_n \cdot \atr_1$  where $(\atr_n,\atr_1) \in \cfo$  is a \serc{} trace of $\aprog_{\pcinstr}$. 
Because $\atr_{n}$ and $\atr_{1}$ are related by $\cfo$, then $\atr_{n}$ is not a right mover because of $\atr_{1}$ and a read-write dependency between the two. Thus, we can relate the two transactions using $\mcfo$. 
\end{proof}

\begin{proof}[Proof of Theorem~\ref{them:MovRobPcSi}]
Similar to before it is enough to show: if $\aprog$ is not robust against \pcc{} relative to \sic{} then we have a simple cycle in the commutativity dependency graph of $\aprog_\pcinstr$ of the form above. Assume $\aprog$ is not robust against \pcc{} relative to \sic{}.
Then, from Theorem \ref{them:RobPcSiInstr}, we obtain that if $\sem{\aprog}$ reaches an error state under \serc{} then we will have the following trace $\atrace$ under \serc: $\atrace = \alpha \cdot \atrrd{\atr_{\instr}}  \cdot \atr_3 \cdot  \beta \cdot \atr_n \cdot \atrwr{\atr_{\instr}}$\footnote{For simplicity, we assume here that after reaching the error state we execute the writes of $\atr_{\instr}$, i.e., $\atrwr{\atr_{\instr}}$.} where $(\atrrd{\atr_{\instr}},\atr_3) \in \cfo$, $(\atr_3,\atr_n) \in \hbo$, $(\atr_n,\atrwr{\atr_{\instr}}) \in \sto$, and we don't have two successive $\cfo$ in the happens before between $\atr_3$ and $\atr_n$. In $\atrace$, $\atrwr{\atr_{\instr}}$ (resp., $\atrrd{\atr_{\instr}}$) represents $\atr_1$ (resp., $\atr_2$) in the theorem statement. 
Note that we may have $\alpha = \beta = \epsilon$ as is the case of the transformed $\mathsf{LU}$ program given in Figure \ref{fig:litmus2Instr}. 
The construction of the cycle in the commutativity dependency graph follows the same procedure taken in the proof of Theorem \ref{them:MovRobCcPc}. The only difference is that for every two transactions of $\atrace$ that are part of the happens before between $\atr_3$ and $\atr_n$, if the two are not connected by either $\po$, $\rfo$, $\sto$, or $\cfo$ then they must be the reads and writes of the same original transaction in $\aprog$. In this case, in the commutativity dependency graph we have the two transactions related by $\sametro$. 
\end{proof}

\end{document}